\documentclass[11pt]{article}
\usepackage{array,multirow}

\newif\ifnotes
\notesfalse

\title{New Codes on High Dimensional Expanders}
\author{Irit Dinur\thanks{Email: \texttt{irit.dinur@weizmann.ac.il}. Supported by ERC grant 772839, and ISF grant 2073/21. Part of the work was done while visiting the Simons Institute for the Theory of Computing.}\\Weizmann Institute of Science
\and Siqi Liu\thanks{Email: \texttt{sliu18@berkeley.edu}}\\UC Berkeley
\and Rachel Yun Zhang\thanks{E-mail:\texttt{rachelyz@mit.edu}. This research was supported in part by DARPA under Agreement No. HR00112020023, an NSF grant CNS-2154149, and NSF Graduate Research Fellowship 2141064.}\\MIT}
\date{\today}

\usepackage{prettyref}
\usepackage[
pagebackref,
colorlinks=true,
urlcolor=blue,
linkcolor=blue,
citecolor=OliveGreen,
]{hyperref}
\usepackage{fullpage}
\usepackage{amssymb}
\usepackage{amsmath}
\usepackage{amsthm}
\usepackage{amsfonts}
\usepackage{bm}
\usepackage{enumitem}
\usepackage{color}
\usepackage{comment}
\usepackage[capitalize]{cleveref}
\usepackage[dvipsnames]{xcolor}
\usepackage{float}
\usepackage[T1]{fontenc}
\usepackage{todonotes}
\usepackage{asymptote}
\usepackage{mdframed}
\usepackage[most]{tcolorbox}
\usepackage{enumitem}
\usepackage{framed}
\usepackage{mdframed}
\usepackage{scrextend}
\usepackage{multirow}
\usepackage{ifthen}
\usepackage{bbm}
\usepackage{frcursive}
\usepackage{kbordermatrix}

\let\pref=\prettyref

\newcommand{\rnote}[1]{\ifnotes $\ll$\textsf{\color{red} Rachel: { #1}}$\gg$ \fi}
\newcommand{\snote}[1]{\ifnotes $\ll$\textsf{\color{magenta} Siqi: { #1}}$\gg$ \fi}
\newcommand{\inote}[1]{\ifnotes $\ll$\textsf{\color{blue} Irit: { #1}}$\gg$ \fi}

\usepackage[T1]{fontenc}

\definecolor{denim}{rgb}{0.08, 0.38, 0.74}
\definecolor{periwinkle}{rgb}{0.6, 0.6, 0.95}
\definecolor{wildblueyonder}{rgb}{0.64, 0.68, 0.82}
\definecolor{wisteria}{rgb}{0.91, 0.72, 1.00}
\definecolor{thistle}{rgb}{0.85, 0.75, 0.85}
\definecolor{byzantium}{rgb}{0.44, 0.16, 0.39}
\definecolor{deeplilac}{rgb}{0.6, 0.33, 0.73}
\definecolor{jazzberryjam}{rgb}{0.55, 0.04, 0.37}

\definecolor{fireenginered}{rgb}{0.81, 0.09, 0.13}
\definecolor{deepcarrotorange}{rgb}{0.91, 0.41, 0.17}
\definecolor{mangotango}{rgb}{1.0, 0.51, 0.26}

\usepackage{hyperref}
\hypersetup{
    colorlinks=true,
    linkcolor=fireenginered,
    filecolor=fireenginered,
    citecolor=deepcarrotorange,
    urlcolor=deepcarrotorange,
}
\usepackage[hyperpageref]{backref}

\newtheorem{theorem}{Theorem}[section]
\newtheorem{lemma}[theorem]{Lemma}
\newtheorem{claim}[theorem]{Claim}

\newtheorem{corollary}[theorem]{Corollary}

\theoremstyle{definition}
\newtheorem{definition}[theorem]{Definition}
\newtheorem{remark}[theorem]{Remark}

\Crefname{theorem}{Theorem}{Theorems}
\Crefname{claim}{Claim}{Claims}
\Crefname{lemma}{Lemma}{Lemmas}
\Crefname{proposition}{Proposition}{Propositions}
\Crefname{corollary}{Corollary}{Corollaries}
\Crefname{definition}{Definition}{Definitions}

\newcommand{\modstar}[1]{\ \text{mod}^*\ #1}

\newcommand{\bbN}{\mathbb{N}}
\newcommand{\bbP}{\mathbb{P}}

\newcommand{\bbF}{\mathbb{F}}

\makeatletter
\newcommand{\customlabel}[2]{%
   \protected@write \@auxout {}{\string \newlabel {#1}{{#2}{\thepage}{#2}{#1}{}} }%
   \hypertarget{#1}{#2}
}
\makeatother

\newcommand\numberthis{\addtocounter{equation}{1}\tag{\theequation}}

\newcounter{casenum}
\newenvironment{caseof}{\setcounter{casenum}{0}}{\vskip.5\baselineskip}
\newcommand{\case}[2]{
    \refstepcounter{casenum}
    \ifthenelse{\equal{\value{casenum}}{0}}{
    \vskip.5\baselineskip\par\noindent
    }{}
    {\it Case \arabic{casenum}:} {\it #1}
    \vskip0.1\baselineskip
    \begin{addmargin}[1.5em]{1em}
    #2
    \end{addmargin}
}

\newcounter{subcasenum}

\newcounter{casenumb}

\newcounter{subcasenumb}

\newcommand\T[1]{{X}_{\scalebox{0.5}{+}#1}}
\def\F{\mathbb{F}}
\newcommand{\bits}{\F_2}
\newcommand\sett[2]{\left\{ #1 \left| \; \vphantom{#1 #2} \right. #2  \right\}}
\newcommand{\set}[1]{\{#1\}}
\newcommand{\Set}[1]{\left\{#1\right\}}
\newcommand{\iprod}[1]{\langle#1\rangle}

\def\L{{\cal L}}
\def\ba{\mathbf{a}}
\def\a{\alpha}

\newcommand{\eqdef}{\stackrel{\mathrm{\Delta}}=}
\def\upperrw{\smallfrown}
\def\lowerrw{\smallsmile}

\newcommand{\Esymb}{\mathbb{E}}
\newcommand{\Psymb}{\mathbb{P}}

\DeclareMathOperator{\dist}{dist}

\DeclareMathOperator*{\E}{\Esymb}

\DeclareMathOperator*{\ProbOp}{\Psymb}

\renewcommand{\Pr}{\ProbOp}
\def\one{{\mathbf{1}}}

\newcommand{\R}{\mathbb R}
\newcommand{\Rnn}{\R_+}
\newcommand{\norm}[1]{\lVert#1\rVert}
\newcommand{\snorm}[1]{\norm{#1}^2}

\def\sM{\tilde M}

\newcommand{\card}[1]{\lvert#1\rvert}
\def\lin{\ell}

\newcommand{\remove}[1]{}

\def\calB{{\cal B}}
\def\calC{{\cal C}}
\def\calG{{\cal G}}
\def\calP{{\cal P}}
\def\calT{{\cal T}}

\newcommand\res[2]{{\rm Res}_{{#2}\leftarrow{#1}}}
\def\FX{{\cal F}X}

\begin{document}

\sloppy
\maketitle
\begin{abstract}

We describe a new parameterized family of symmetric error-correcting codes with low-density parity-check matrices (LDPC). 

Our codes can be described in two seemingly different ways. First, in relation to Reed-Muller codes: our codes are functions on a subset of $\F^n$ whose restrictions to a prescribed set of affine lines has low degree. Alternatively, they are Tanner codes on high dimensional expanders, where the coordinates of the codeword correspond to triangles of a $2$-dimensional expander, such that around every edge the local view forms a Reed-Solomon codeword. 

For some range of parameters our codes are provably locally testable, and their dimension is some fixed power of the block length. For another range of parameters our codes have distance and dimension that are both linear in the block length, but we do not know if they are locally testable. The codes also have the {\em multiplication property}: the coordinate-wise product of two codewords is a codeword in a related code.

The definition of the codes relies on the construction of a specific family of simplicial complexes which is a slight variant on the coset complexes of Kaufman and Oppenheim. We show a novel way to embed the triangles of these complexes into $\F^n$, with the property that links of edges embed as affine lines in $\F^n$.

We rely on this embedding to lower bound the rate of these codes in a way that avoids constraint-counting and thereby achieves non-trivial rate even when the local codes themselves have arbitrarily small rate, and in particular below $1/2$. 

\end{abstract}
\thispagestyle{empty}
\newpage

\tableofcontents
\pagenumbering{roman}
\newpage
\pagenumbering{arabic}

\section{Introduction}

A locally testable code (LTC) is an error correcting code that has a property-tester. The tester reads $q$ bits that are randomly chosen, and rejects words with probability that grows with their distance from the code. 

LTCs were initially studied and constructed side by side with PCPs (probabilistically checkable proofs). It was only recently that the question of existence of LTCs with the $c^3$ property was resolved in the affirmative \cite{DinurELLM2022, PanteleevK22}. A code has the $c^3$ property if it has constant rate, constant distance and testable with constant locality. 

It was initially hoped that high dimensional expanders, a la \cite{LubotzkySV2005a, LubotzkySV2005b} can serve as the combinatorial structure underlying the code, and all one needs to do is to find an appropriate collection of local codes (at the links) to match up with the combinatorics, see \cite{DinurK2017, DinurHKR2018, DiksteinDHR20}. 

This approach has not borne fruit up until now, essentially due to the stringent requirements on the local codes, which turned out difficult to fulfill. Nevertheless, $c^3$ codes were eventually constructed by circumventing the problem and switching from simplicial to square complexes. The main benefit of square complexes is that their local views support tensor codes, which satisfy the requirements for local testability.

In this work we return to the simplicial setting and construct a new parameterized family of locally testable codes on simplicial (bounded-degree) high dimensional expanders. 
In addition to serving as a new and potentially interesting family of LDPC error-correcting codes,  these codes satisfy 
further properties that could be potentially useful for other applications such as PCPs and beyond. In particular, the codes are symmetric in the sense of \cite{KaufmanW16, KaufmanL12}, meaning that there is a group acting on the coordinates of the codeword, that takes every coordinate to every coordinate. In addition, the fact that local views of our codes are Reed-Solomon, immediately implies that
the codes satisfy the {\em multiplication property}: the coordinate-wise product of two codewords is a codeword in a related code. 


Our codes can be described in two seemingly different ways. First, in relation to Reed-Muller codes: our codes are functions on a subset  $\bar X_n\subset \F^n$ whose restrictions to a prescribed set of affine lines has low degree. Alternatively, they are Tanner codes on high dimensional expanders, where the coordinates of the codeword correspond to triangles of a $2$-dimensional expander $X$, such that around every edge the local view forms a Reed-Solomon codeword. 
The definition of the codes relies on the construction of a specific family of simplicial complexes whose triangles embed naturally into $\F^n$, with the property that links of edges embed as affine lines\footnote{An affine line in $\F^n$ is  given by $\ba_0\in \F^n$ and $\ba_1\in \F^n\setminus\set 0$ so that $\ell_e = \sett{\ba_0 +t\ba_1}{t\in \F}$} in $\F^n$.
\begin{theorem}\label{thm:mainX}
    Let $\F=\F_q$ be a fixed finite field. For every $n$ divisible by $9$ there exists a connected $2$-dimensional $3$-partite simplicial complex $X_{n}$, such that 
    \begin{enumerate}
         \item[(a)]\label{item:mainX-link} For each vertex $v\in X_n(0)$, the link of $v$ is a bipartite $q$-regular graph on $2q^2$ vertices whose normalized adjacency matrix has second largest eigenvalue $\lambda_2=1/\sqrt{q}$.
        \item[(b)]\label{itm:size} There is a set of points $\bar X_n \subset \F^n$ and an injective map $\iota:X_n(2)\to \F^n$, such that $\bar X_n=Im(\iota) \subset \F^n$, and \[\card{X_n(2)} = \card{\bar X_n} \geq q^{ c n}\] for some absolute constant $c>0$.
        \item[(c)] \label{item:mainX-edge} There is a set $\L_n$ of affine lines in $\F^n$ with one line per edge $e\in X_n(1)$. The line corresponding to an edge $e$ is given by a bijection $\lin_e:\F\to \T e$ such that $\iota\circ\lin_e:\F\to\F^n$ is an affine line in $\F^n$, where we denote $\T e = \sett{t\in X_n(2)}{t\supset e}$. Let $\L_n=\sett{\iota\circ\lin_e}{e\in X(1)}$. 
        \item[(d)] $X_n$ is $3$-partite, so the edges have $3$ distinct types and we denote by $\L^i_n$ the set of lines corresponding to type-$i$ edges.
   \end{enumerate}
    Furthermore, $X_n$ and the maps $\iota,\set{\lin_e}_e$ are constructible in polynomial time in the size of the complex $X_n$.
    \end{theorem}
The complexes $\set{X_n}$ are a slight variant on the coset complexes of \cite{KaufmanO181}. 
The embedding of $X_n(2)$ into $\F^n$, so that the links of edges map into affine lines is new, and allows us to define a new family of codes on $X_n$,
\begin{definition}[HDX local Reed-Solomon Codes]\label{def:code}
    Let $\F=\F_q$ and $\set{X_n}_n$ be as above, and let $d_1,d_2,d_3 < q$ be positive integers. We define a family of codes as $n\to\infty$, 
    \begin{align}\label{eq:RMcode}
    C_{n,d_1,d_2,d_3} &= \sett{f:\bar X_n\to\F}{ \forall i=1,2,3,\,\forall\lin\in \L_n^i,\quad f \circ \lin \in RS(q,d_i)}
    \end{align}
    where $RS(q,d)\subset \F^{\F}$ is the Reed Solomon code of degree $d$ over $\F$. Note that we have three distinct degree parameters since the edges of the complex have three distinct types.
    
    The code $C_{n,d_1,d_2,d_3}$ can alternatively be described as 
    \begin{align}\label{eq:Tannercode}
        C_{n,d_1,d_2,d_3} &= \sett{f:X_n(2)\to\F}{ \forall e\in X(1),\; f|_{\T e} \in C_e}. 
    \end{align}
for an appropriate choice of codes $C_e\subset \F^{\T e}$ such that $C_e$ isomorphic to $RS(q,d_i)$ whenever $e$ is an edge of type $i$.

We also denote $C_{n,d}=C_{n,d,d,d}$. In the sequel we will often focus on $C_{n,d}$ as it contains all of the ideas, but we wanted to include the slightly more general definition of $C_{n,d_1,d_2,d_3}$.
\end{definition}
The isomorphism between definitions \eqref{eq:RMcode} and \eqref{eq:Tannercode} is essentially given by $\iota$ from \pref{thm:mainX}. 
The map $\iota$ identifies $X(2)$ and $\bar X_n\subset \F^n$, and further identifies the set of triangles containing an edge $e\in X(1)$ with an affine line in $\F^n$ (see also Section \ref{subsec:globalcode} 
and Claim \ref{claim:iso} for a full proof). 

By definition \eqref{eq:RMcode} we can see that every $n$-variate polynomial of total degree at most $d=\min(d_1,d_2,d_3)$ gives rise to a codeword in $C_{n,d}$. This allows us to give a non-trivial lower bound on the dimension of $C_{n,d}$ even when $d$ is small and constraint-counting fails. On the other hand, definition \eqref{eq:Tannercode} allows us to prove distance and local testability through the high dimensional expansion machinery. 

\begin{theorem}\label{thm:mainCode}        
    \remove{so the local rate stuff only holds for prime field size, and i think that gets used in local testability} 
    Fix a prime power $q$ and let $d<q$. The family $\set{C_{n,d}}_n$ has the following properties
\begin{enumerate}
    \item \label{item:thm1.3-rate} \textbf{Rate:} The dimension of the code is at least the dimension of the Reed-Muller code with $n/3$ variables and degree $d$. Furthermore, for $d>\frac 2 3 q$ the code has dimension $\Omega(\card{X_n(2)})$ (namely, linear rate as $n\to\infty$).
    \item \label{item:thm1.3-distance} \textbf{Distance:} If $d < q- \Omega(\sqrt{q})$ the code has constant relative distance $\Omega_q(1)$. 
    \item \label{item:thm1.3-testability} \textbf{LDPC and Local Testability:} 
    For all $d$ and $n$, the code is defined by parity checks of length $d+2$. When $q$ is prime, if $d<q/4$ the code is locally testable with $d+2$ queries. 
    \item \label{item:thm1.3-multiplication} \textbf{Multiplication:} For any $d_1,d_2$, and for any  $w_1\in C_{n,d_1}$ and $w_2\in C_{n,d_2}$, we have $w_1\odot w_2\in C_{n,d_1+d_2}$ where $w_1\odot w_2$ is the coordinate-wise product of $w_1$ and $w_2$.
    \item \label{item:thm1.3-transitivity} \textbf{Symmetric:} There is a transitive group action on the coordinates of the code that preserves the code. Transitivity means that for every pair of coordinates $t,t'$ there is a group element that takes $t$ to $t'$.
\end{enumerate}
\end{theorem}

\begin{remark}
Here is a table that summarizes the rate, relative distance, and local testability of $\set{C_{n,d}}_n$ in different regimes of $d$.

\begin{center}
\begin{tabular}{ |p{3cm} ||p{2cm} |p{1.5cm}|p{3cm}|  }
 \hline
 $d$ & Rate  & Distance & Local testability\\
 \hline 
 $(0, q/4)$   & \multirow{2}{2cm}{$\ge{{n/3}\choose{d}}$}   & \multirow{3}{1.5cm}{$\Omega_q(1)$} &   $d+2$\\[0.5ex] 
 $[q/4, \frac{2}{3}q)$ &     &   & \multirow{2}{3cm}{?}\\[0.5ex] 
 $[\frac{2}{3}q, q -\Omega(\sqrt{q}))$ &$\Omega(\card{X_n(2)})$ &   &  \\[0.5ex] 
 \hline
\end{tabular}
\end{center}
We fail to give a family of $c^3$ codes (with constant relative rate, constant relative distance, and local testability with constant number of queries). Specifically, in the low degree regime $d< q/4$, the lower bound on the relative rate is subconstant $>\frac{{{n/3}\choose{d}}}{\card{X_n(2)}} = q^{-\Omega(n)}$. In the high degree regime $d > \frac{2}{3}q$, our approach for showing local testability fails, and we don't know whether or not local testability occurs.
\end{remark}
%
%
There is a natural way to generalize our complexes and codes to higher dimensions. In \pref{sec:higher} we describe this generalization and show that distance as well as local testability ``trickles down''. We also show that the codes have a homological description as the space of cycles on the top dimension. 
\subsection{Background}
Recent works \cite{DinurELLM2022,PanteleevK22} give locally testable codes that have the $c^3$ property, namely they have constant relative rate, constant relative distance and locally testable with a constant number of queries. These codes are constructed on an especially designed squares complex. Similar schemes have been hypothesized to exist on simplicial complexes (see discussion in \cite{DinurELLM2022}, and see \cite{firstK2022good}). 

The simplicial complex codes are defined as 
\begin{equation}\label{eq:HDXcode} C = \sett{w\in \F^{X(2)}}{w|_{\T e}\in C_e}
\end{equation}
where we fix, for each edge $e$, a local code $C_e \subseteq \F^{\T e}$ where $\T e$ is the set of triangles containing the edge $e$. In other words, $C$ is defined by aggregating the local codes $C_e$ into one big parity check matrix, using the structure of the HDX. The definition is simple, but the challenge is in analysing properties of $C$ such as rate, distance, or local testability.
\begin{enumerate}
    \item Distance follows quite directly from the  distance of each $C_e$ together with expansion of $X$.
    \item Local testability can be shown if we manage to show a local version of local testability, per each vertex $v$. Namely, we look at $C_v := C|_{\T v}\subseteq \F^{\T v}$, the restriction of the code to the set $\T v$ of triangles containing a vertex $v$. This itself is a linear subspace of $\F^{\T v}$. We say that $C_v$ itself is locally testable if a good test for  $z\stackrel ?\in C_v$ is to select a random $e\ni v$ and check if $z|_{\T e}\in C_e$.  
    
    If each $C_v$ is locally testable, then the high dimensional expansion of $X$ would allow us to deduce local testability of the entire code $C$. 

    This is the content of \pref{thm:loc2glob}, whose proof is based on the same ideas underlying the proof of local testability in the the squares complex \cite{DinurELLM2022,PanteleevK22} and also earlier the proof of cosystolic expansion in high dimensional expanders \cite{KaufmanKL2014,EvraK2016}, and is similar to \cite[Proposition 6.4]{firstK2022good}, although there are technical differences.\footnote{They show that if the links are coboundary expanders, then the global sheaf is a cosystolic expander. This is morally equivalent to saying that local robustness implies global local-testability. However, our proof makes use of a weaker local robustness condition, which is what we manage to prove for the local codes.}
    \item The dimension (or rate) of $C$ could a priori be $0$, so one needs to show that $C$ is non-trivial somehow. So far the only successful method has been through constraint counting. One shows that the total number of parity checks is smaller than the number of degrees of freedom. Of course the constraints are often highly dependent so this argument is not tight, and one would like new techniques for lower bounding the dimension.  
\end{enumerate}
A major difficulty in construction of LTCs is in coming up with a collection of local codes $C_e$ that simultaneously allow proving that each $C_v$ is locally testable, and at the same time maintain non-trivial dimension for $C$.
Indeed, several prior works gave a general framework for HDX codes  but did not know how to instantiate them non-trivially. In \cite{DDHR}, the authors gave a definition based on double samplers, \cite{DELLM} constructed codes on squares complex but also considered codes on the LSV complex with no non-trivial rate, and \cite{firstK2022good} defined sheaves on high dimensional expanders. 

In this work we give a first family of codes on high dimensional expanders that are defined by instantiating \eqref{eq:HDXcode}. As mentioned before, our local codes $C_e$ are Reed-Solomon codes. The codes $C_v$ turn out to be the following
\begin{equation}\label{eq:local-codes}
    C_v\cong C_{d_x,d_y} = \sett{f:\F_q^3\to \F_q}{ \forall a,b,c,\,\, deg_x(f(x,b,c))\leq d_x; \,\, deg_y (f(a,y,ay+c)) \leq d_y}.     
    \end{equation}
We prove that $C_{d_x,d_y}$ is agreement-testable when $d_x + d_y < q/2$ (see \pref{thm:local-testability}) and when $q$ is a prime. Namely, given two functions $X,Y:\F_q^3\to\F_q$ such that
for all $b,c\in \F_q$, $X(x,b,c)$ is a low degree  function of $x$; and such that for all $a,c\in \F_q$, $Y(a,y,ay+c)$ is a low degree  function of $y$, if
\[ \Pr_{a,b,c}[X(a,b,c) \neq Y(a,b,c)] < \epsilon\]
then there is some $f\in C_{d_x,d_y}$ that is close to both $X$ and $Y$.
To show that this is a good test we adapt the analysis of Polyschuk and Spielman \cite{PolishchukS94} of the local testability of bivariate polynomial functions.  
We remark that the local testability here is slightly weaker: there is a quadratic, rather than linear, relation between the distance of the word to a code and the probability of failing the test. Nevertheless we derive local-testability of $C$ from this slightly weaker local testability of the $C_v$'s. 

Finally, as for the dimension of the code $C$. For degree parameters that allow local testability, constraint counting is out of the question, because there are more constraints than degrees of freedom. Instead we show that the code contains many codewords by showing that every low degree multi-variate polynomial is a valid codeword in of our code. We leave for future research to try and get tighter bounds on the rate of these codes.

\paragraph{Multiplication Property} A triple $C_1,C_2,C_3$ of codes are so-called {\em multiplication codes} if for any $w_1\in C_1$ and $w_2\in C_2$ the word $w_1\odot w_2$ obtained by coordinate-wise product satisfies $w_1\odot w_2\in C_3$. This property was introduced in \cite{Meir13} who used it towards a proof for $IP=PSPACE$ involving abstract codes, not necessarily Reed-Muller. The original proof of \cite{LFKN,Sha:IP:PSPACE} relied on
Reed-Muller codes which are indeed multiplication codes. The multiplication property is used for example in the proof of the PCP theorem \cite{LFKN,AS,ALMSS}, as well as for further recent cryptographic applications. For such applications it would be very useful to have linear-time encodable multiplication codes.

Since our codes are defined as lifts of Reed-Solomon codes they are automatically multiplication codes for appropriate degree choices. By choosing the parameters appropriately, we can get $C_1,C_3$ an LDPC code with linear rate, and $C_2$ a locally testable (LDPC) code with polynomial rate, such that this triple has the multiplication property (see Lemma \ref{lem:mult}).  

\paragraph{Connection to Reed-Muller codes, lifted codes, and Tanner Codes.} 
Recall that a Tanner code is given by two pieces of data:
\begin{itemize}
    \item A bipartite graph $\calG=(\calB,\calP,E)$ with left vertices $\calB$ (for bits), right vertices $\calP$ (for parity checks), and edges $E$. 
    \item For each $p\in \calP$ there is a local code $C_p \subset \F^{\Gamma(p)}$ where $\Gamma(p)\subset \calB$ are the neighbors of $p$.
\end{itemize}
Given $\calG$ and $\set{C_p}$, the Tanner code is  
\[ \calT(\calG, \set{C_p}) = \sett{w\in \F^\calB}{\forall p\in\calP,\; w|_{\Gamma(p)}\in C_p}.\]

The Reed-Muller code $RM_{n,d}$ is the space of all $n$-variate polynomials of total degree at most $d$. 
For a broad set of parameters, Reed-Muller codes are themselves known to be Tanner codes, with local parity check codes being Reed Solomon codes, and with a $\calB = \F^n$ and $\calP$ being the set of all affine lines in $\F^n$, so that the bipartite graph $\calG_{RM}$ is the point-vs.-line graph which connects a point in $\F^n$ to all lines passing through it. 
For some parameter regimes (when the degree is high) such Tanner codes form a larger space than the Reed-Muller codes, as was discovered and studied in \cite{GuoKS13}, where such codes are  termed ``lifted codes''.
Later in \cite{FGW2017}, partial lifts are considered where some of the vertices of $\calP$ are erased and one considers the remaining punctured code. This resembles the situation in our codes, as we explain next.

The code $C_{n,d_1,d_2,d_3}$, as per \eqref{eq:RMcode}, can be viewed a Tanner code whose graph is a subgraph of $\calG_{RM}$ obtained by keeping a subset $\L_n$ of the lines, and a subset $\bar X_n\subset \F^n$ of the points. We also allow different degree restrictions on different types of lines. If we set $d=d_1=d_2=d_3$ the parity check matrix of $C_{n,d}$ is thus a sub-matrix of the parity check matrix of the Reed-Muller code matrix.
As for \eqref{eq:Tannercode}, it can be seen as a Tanner code on the bipartite graph $\calG=(X(2),X(1),E)$ connecting each edge $e\in X(1)$ on the right to the set $\T e \subset X(2)$ of triangles containing it, and putting a Reed-Solomon code on every $\T e$ in a specified way. 
\paragraph{Two-query testability}
One can often convert a code that is testable with $q$ queries to another code that is testable with two queries, while increasing the alphabet. This is done by converting each codeword to the list of its local views. For example, in the case of Reed-Muller codes, instead of representing a polynomial by its evaluation on points, we represent it by providing its restriction to each plane. This is sometimes called the planes table, and is two-query testable by the plane-vs.-plane test, see \cite{RaSa}. In the case of $c^3$ LTCs of \cite{DinurELLM2022, PanteleevK22}, this conversion is immediate, although it is not described there explicitly. For our codes it also holds, and we include details in Claim \ref{claim:2q}. This property was also highlighted in \cite{firstK2022good}.

\subsection{Further Work}
This work raises many questions for further investigation. 
\begin{itemize}
    \item {\bf Better bounds.} What is the exact rate of our codes? We give a lower bound based on Reed-Muller codes, but we don't know how close this bound is to the truth. Are there parameter regimes where our codes are $c^3$ ? 
    \item {\bf Efficient encoding.} Our codes are specified by their parity-check matrices. A generator matrix can be computed in $O(N^3)$ time, where $N$ is the block length. Given the generator matrix, computing an encoding takes time $O(N^2)$. Can encoding be done more efficiently, ideally even in linear-time?
    \item {\bf Quantum codes.} It has recently become clear that locally testable HDX codes are related to quantum LDPC codes of the CSS type. Our construction gives a $2$-chain (from vertices to edges to triangles) which can automatically be viewed as a quantum LDPC code. Indeed, we show in Section \ref{subsec:sheaves} how any HDX code gives rise to a sheaf which, by \cite[Section 7.4]{firstK2022good}, gives rise to a quantum code.  Our local testability analysis implies distance in one direction, whereas for a good quantum LDPC code, one needs to prove both distance as well as co-distance. This is an interesting challenge. Moreover, in Section \ref{sec:higher} we show how our construction generalizes to $k$ dimensions, which, by the same recipe as above, can be converted to a sheaf and a $k+1$-chain. Studying the resulting quantum codes at intermediate levels $0<i<k$ seems like an interesting direction.  

    \item {\bf Sparsified Grassmannian.} We have constructed a collection of affine lines with the property that each point has exactly three lines passing through it. This is a very sparse collection of lines, that never-the-less expands quite well when we walk from point to line to point. How does this generalize to higher dimensional complexes? The question of finding sparse Grassmannians has been studied in \cite{MR-LDT, DiksteinDFH2018, KaufmanT23, Gol23} and could provide new PCP gadgets with properties similar to \cite{KMS}.
    
    \item {\bf Other coset complexes.} We have considered a variant of the coset complexes of Kaufman and Oppenheim, by choosing subgroups based on a sub-ring of the ring chosen in \cite{KaufmanO181}. Many high dimensional expanders are based on matrix groups \cite{LubotzkySV2005a, LubotzkySV2005b, KaufmanO181, ODonnellP2022}. Thus similar embddings of the complexes into $\F^n$ could potentially be useful for coming up with new codes, whose properties can be studied through the high dimensional expansion framework.

    \item {\bf Generalized sumcheck.} Many interactive proof systems reduce proving certain relations to checking that a multi-variate polynomial $p(x_1,\dots,x_m)$ over $\F^m$ sums up to zero over some subset $\mathbb{H}^m \subseteq \F^m$. The sumcheck protocol is designed to check this relation by recursively restricting variables of $p$ and checking the condition over smaller sets of random variables. 
    Can one design sumcheck-like protocols for codes other than Reed-Muller? For example, in the case of our code, we would need to find subsets of points that play the role of $H^m$, so that one can test whether some partial sums of a codeword sum up to zero. 
    
    Perhaps a first step would be to find a sumcheck protocol for our local codes $C_{d_x,d_y}$. \inote{I rewrote the above. I hope I am not missing any point you tried to make here}

\end{itemize}

\subsection{Organization}

\rnote{rewrote this section to make it a bit more readable (like have words and not just references), but could probably still be improved if we care at all}
We define our coset complex in Section~\ref{sec:complex} and prove its properties (\pref{thm:mainX}) in Sections \ref{subsec:X}, \ref{subsec:links}, and \ref{subsec:embed}. We then define our code on this coset complex in Section~\ref{subsec:globalcode}. Its properties as detailed in Theorem~\ref{thm:mainCode} are proved in the following sections:
\begin{itemize}
    \item In Section~\ref{subsec:mult} we show that our code has the multiplication property, and in Section~\ref{subsec:distance}, we show that our code has constant distance in appropriate parameter regimes (items~\ref{item:thm1.3-distance} and~\ref{item:thm1.3-multiplication} of Theorem~\ref{thm:mainCode}). We also show that our code has the transivity property (item~\ref{item:thm1.3-transitivity}).
    \item We discuss the rate (item~\ref{item:thm1.3-rate}) of the global and local codes in Sections~\ref{sec:rate-global} and~\ref{sec:rate-local} respectively.
    \item In Sections~\ref{sec:locLTC} and~\ref{sec:globLTC}, we prove that the local and global codes are agreement testable.
\end{itemize}
Finally, in Section \ref{sec:higher} we describe the generalization to higher dimensions.

\section{Preliminaries}\label{sec:prelim}

\subsection{Expander Graphs}
A $d$-regular graph $G$ is said to be a $\gamma$-one-sided expander if it has eigenvalues $d=\lambda_1 \ge \lambda_2 \ge ... \ge \lambda_n \geq -d$ which satisfy $\lambda_i \leq \gamma\cdot d$ for all $i>1$.

\begin{lemma}[Alon-Chung]\label{lemma:AC}
Let $G=(V,E)$ be a $d$-regular $\gamma$ one-sided expander. Let $T\subseteq V$ be such that the graph induced on $T$ has average degree at least $\delta d$. Then $|T| \ge (\delta - \gamma )\cdot |V|$.
\end{lemma}
\begin{proof}
Let $A$ be the normalized adjacency matrix of $G$ and let $f$ be the indicator function of $T$. Using the spectral decomposition $f = \frac {|T|}{|V|} {\mathbf 1} + f^\perp$ we get
\[\delta d |T| \leq 2E(T) = f^\top A f \leq \card T^2 d/|V| + \gamma d \card T
\] 
where $E(T)$ denotes the number of edges in the induced graph on $T$.
Dividing both sides by $d|T|$ and rearranging gives the lemma.
\end{proof}

\subsection{High Dimensional Expanders}
A pure \(k\)-dimensional simplicial complex \(X\) is a set system (or hypergraph) consisting of a set of vertices $X(0)$ and an arbitrary collection of subsets of size \(k + 1\) together with all their subsets. The sets of size \(i+1\) in \(X\) are denoted by \(X(i)\). We will sometimes omit set brackets and write for example \(uvw\in X(2)\) instead of \(\set{u,v,w}\in X(2)\). As convention \(X(-1) = \set{\emptyset}\). Unless it is otherwise stated, we always assume that \(X\) is finite.

Let \(i < k\) and \(s \in X(i)\). It is standard to define the link of \(s\) to be a \(k-i-1\)-dimensional simplicial complex defined by
\(X_s = \sett{t \setminus s}{t \in X, t \supseteq s}\). We also define the less-standard but useful notation 
\begin{definition}[Star]
    For a $k$-dimensional complex $X$ and a face $s\in X(i)$ for some $i<k$, the star of $s$ is the $k$-dimensional complex containing all faces that contain $s$.
\[\T s(j) = \sett{t \in X(j)}{t \supseteq s}.\]
\end{definition}
For a face $s\in X(i)$, there is a natural bijection $\T s(j) \to X_s(j-i-1)$ mapping $t\in \T s$ to $t\setminus s \in X_s$.
\begin{definition}[High dimensional local spectral expander]
    Let \(X\) be a \(k\)-dimensional simplicial complex. Let \(\lambda \geq 0\). We say that \(X\) is a \(\lambda\)-one sided local spectral expander if for every \(s \in X^{\leq k-2}\), the graph \((X_s(0),X_s(1))\) is a \(\lambda\)-one sided spectral expansion.
\end{definition}

\begin{definition}[Coset complex]\label{def:cosetcomplex}
A $k$-dimensional {\em coset complex} is given by a group $G$ and subgroups $K_1,\ldots,K_{k+1}$. The vertices are all cosets of $K_i$, and the $i$-faces are all $i+1$-tuples of cosets that have a non-empty intersection. The complex is denoted 
$X[G;K_1,\ldots,K_{k+1}]$. 
\end{definition}
A beautiful construction of a constant-degree coset complex that is a high dimensional expander was given in \cite{KaufmanO181}, see also \cite{ Harsha_Saptharishi_2019,ODonnellP2022}.

It is not hard to see that links in a coset complex are themselves coset complexes.

\subsection{Random Walks on High Dimensional Expanders}
Let $X$ be a regular two-dimensional complex, so that every vertex touches the same number edges, and every edge touches the same number of triangles. This assumption is not needed, but it is satisfied by our coset complexes and it slightly simplifies the definitions below.

Let $V=X(0), E = X(1), \T{}=X(2)$, and also denote $X(-1) = \set \phi$. 

The down operator $D^i:\R^{X(i)}\to \R^{X(i-1)}$ (for $0\leq i\leq 2$) and the up operator $U^i:\R^{X(i)}\to\R^{X(i+1)}$ (for $-1\leq i\leq 1$) are defined by 
\begin{align*}
\forall a\in X(i-1),\quad D^i f (a) &= \E_{b\supset a} [f(b)],\\
\forall b\in X(i+1),\quad  U^i g (b) &= \E_{a\subset b} [g(a)].
\end{align*} 
We will often drop the superscript $i$ in $U^i$ and $D^i$ when it is clear what $i$ is.

Let $e\lowerrw e'$ denote the lower random walk on the edges (choose a random $e$, then $v\in e$, then $e'\ni v$). It is easy to see that the Markov operator corresponding to this walk is just $UD$.
Let $e\upperrw e'$ denote the non-lazy upper random walk on the edges (choose a random $e$, then $t\supset e$, then $e'\in t$ such that $e'\neq e$). It is not hard to see that the Markov operator corresponding to this walk, denoted $M^+$, satisfies $DU = \frac 2 3 M^+ + \frac 1 3 I$. 

We define inner products on the spaces $\R^V,\R^E$ by expectation according to the uniform distribution (here we are using the regularity assumption). 

For example, for $f,g:V\to \R$
\[\iprod{f,g} = \E_{x\in V} [f(x)g(x)],\qquad \norm f = (\iprod{f,f})^{1/2}.\]
The following is by now well known, see \cite{DinurK2017,KaufmanO-RW20}, and we include a proof for completeness.

\begin{lemma}\label{lemma:updown}
Let $X$ be a one-sided $\gamma$-link expander. Every $g\in \R^E$ satisfies
\[\iprod{g,M^+g} \leq \iprod{g,(UD + \gamma I) g} .
\]
\end{lemma}
In particular, for a set $R\subset X(1)$ such that $\beta= \Pr_{e\upperrw e' } [ e \in R | e'\in R ]$ it must be that
\[ \Pr_{ e\lowerrw e' } [ e \in R | e'\in R ] \ge\beta - \gamma\]
by applying the lemma on $g = \one_R$ and observing that $\Pr_{t\in X(2), e\neq e'\subset t}[e,e'\in R]=\beta\Pr[R]$. 
\begin{proof}
\begin{align*}\label{eq:verta}
\iprod{g,M^+g} &= \E_{abc \in \T{}} g(ab)g(ac) \\
&=\E_a\E_{bc\in X_a(1)}[g(ab)g(ac)] \\
&\le \E_a\E_{b,c\in X_a(0)}[g(ab)g(ac)]+\gamma\E_a\E_{b\in X_a(0)}[g(ab)^2]\\
&= \iprod{g,UD g}  + \gamma\norm g^2
\end{align*}  
where the inequality follows since $X_a$ is a $\gamma$-one-sided expander for each $a$, and this means that for any $f:X_a(0)\to \R$ (and in particular setting $f(b) := g(ab)$),
\[ \E_{bc\in X_a(1)} f(b)f(c) \leq \E_{b,c\in X_a(0)}f(b)f(c) + \gamma\cdot\E_{b\in X_a(0)} f(b)^2.
\]
Note that the distribution of choosing $a$ and then two edges $ab,ac$ independently, is equivalent to choosing a random edge $e'$, then a random vertex $a\in e'$ and then another edge $e\ni a$. 
\end{proof}
The following swap walk $S_{0,1}$ that starts from a random edge in $X(1)$ and ends at some vertex in $X(0)$ will be used in the analysis of local testability of the global code. Starting with a random edge $e$, choose a random triangle $t \ni e$ and output $v = t \setminus e$.
\begin{lemma}\label{lemma:sM}
Let $X$ be a two-dimensional $\gamma$-link-expander. The random walk $\tilde{M} = S_{0,1}D$ on the edges has second largest eigenvalue bounded by $3\gamma$.
\end{lemma}
\begin{proof}
We claim that 
\[  M^+ UD = \frac 1 2 S_{0,1}D + \frac 1 2 UD.\]
Let us analyze the random walk corresponding to $M^+ UD$. 
We start from an edge $e$, go (via $M^+$) to a random edge $e_1$ such that $e\cup e_1\in X(2)$. We then 
 go from $e_1$, via $UD$, to an edge $e'$. In this step two things can happen, each with probability $1/2$. Either $e'$ doesn't contain $\set v = e_1\cap e$, in which case we end up with the distribution of $S_{0,1}D$; or $e'\ni v$, in which case we end up with $UD$. This proves the required equality.
%
%
We can thus write $S_{0,1}D = 2M^+UD-UD$. Furthermore, by \cref{lemma:updown} we have that $\norm{2M^+UD-UD}  \leq \norm{2(UDUD+\gamma UD)-UD}$. So for any $f\in\R^E$ with $\E[f]=0$ we get,
\[ \iprod{f,S_{0,1}D f} \leq \iprod{ f,(2UDUD -  UD)f} + 2\gamma\norm f^2 = \iprod{ f,U(2DU - I)Df} + 2\gamma\norm f^2  \leq 3\gamma \norm f^2 \]
where in the last inequality we have used the fact that $2 DU-I$ is nothing other than the random walk from vertex to vertex in the graph $(X(0),X(1))$, so by assumption on $X$, $\iprod{(2DU - I)g,g}\leq \gamma\norm g^2$ for every function $g\in \R^V$ such that $\E[g]=0$.
\end{proof}

\subsection{HDX Codes}\label{sec:HDXcode}
An HDX code is defined by two objects 
\begin{itemize}
    \item A $k$-dimensional simplicial complex $X$, and a dimension $0<i\leq k$.
    \item A collection $\set{C_s}$ of local codes $C_s \subseteq \bits^{\T s(k)}$, one per face $s\in X(k-1)$.
\end{itemize} 
The HDX code at dimension $k$ is defined as 
\begin{equation}\label{eq:Tanner}
    \calC^k[X,\set{C_s}] = \sett{ f\in \bits^{X(k)} }{ \forall s\in X(k-1),\; ( \;f(s\cup v) : v\in X_s(0)\;)\in C_s}
\end{equation}
When $X$ is one dimensional these are the expander codes of \cite{SipserS96}.

\paragraph{Example 1: expander code}
An expander code is given by an expander graph $X=(V,E)$ and a local code $C_v$ for every vertex $v\in V$ such that $C_v\subseteq \bits^{\set{e\ni v}}$. A word $f\in \bits^E$ is in the code if for every vertex $v$, the bits on the edges touching $v$ form a codeword in a small code $C_v$. Formally, if $(f(e): e\ni v)\in C_v$. Often the graph $(V,E)$ is $d$-regular and the local codes are taken to all be copies of some $C_0\subset\bits^d$. 

This is a special case of \eqref{eq:Tanner} for dimension $k=1$, where $X(0)=V$ and $X(1)=E$ and $C_v\subseteq \bits^{X_v(0)}$ via the identification $X_v(0)\leftrightarrow \set{e\ni v}$.

\paragraph{Example 2: cocycle codes}
Fix some simplicial complex $X$, and some dimension $0<k<dim(X)$. Suppose for every $s\in X(k-1)$, the local code $C_s$ is taken to be the parity code consisting of all even length words, namely the local code at $s\in X(k-1)$ is,
\[ Z_s = \sett{ f\in \bits^{X_s(0)}}{\sum_{v \in X_s(0)}f(s\cup v)=0}.\]
Then the $k$ dimensional HDX code $C[X,\set{Z_s}]$ coincides with the space of $k$-cocycles of $X$. For example, when $k=1$, the code is spanned by all closed walks.
\subsection{Local-Testability and Agreement-Testability}\label{sec:agr}
For a function $f:A\to B$ and a subset $A'\subset A$ we denote by $f|_{A'}:A'\to B$ the restriction of $f$ to $A'$.

Let $X$ be a $k$-dimensional simplicial complex and assume that for every $v\in X(0)$ we are given a local code $C_v\subseteq \F^{\T v}$. Let $C=\calC^k[X,\set{C_v}]\subseteq \F^{X(k)}$ be an HDX code. 

\begin{definition}[Locally testable code]\label{def:LTC-classic}
    Let $\rho:\Rnn\to \Rnn$ be a strictly increasing function with $\rho(0)=0$, and let $\epsilon>0$. A code $C\subseteq \bits^n$ is a $(Q,\epsilon,\rho(\cdot))$-locally testable code if there is a randomized tester that, upon receiving a given word $f\in\bits^{n}$, queries $f$ in at most $Q$ locations and then accepts or rejects, such that
    if $p=\Pr[\hbox{Tester rejects }f]\leq \epsilon$ then $\dist(f,C) \leq \rho( p )$.
    \end{definition}
We specialize this definition to the case of HDX codes by considering a ``cannonical'' local tester that selects a random vertex $v$ and checks if the restriction of the codeword to the star of $v$ is in the local code $C_v$. Namely, if $f|_{\T v}\in C_v$. The number of queries made by this tester is equal to the maximal number of $k$-faces containing a vertex $v$.   
\begin{definition}[Local testability of HDX Codes]\label{def:LTC}
    Let $\rho:\Rnn\to \Rnn$ be a strictly increasing function with $\rho(0)=0$, and let $\epsilon>0$.
    The HDX code $C=\calC^k[X,\set{C_v}]$ is 
    $(\epsilon,\rho(\cdot))$-locally testable if for any $f\in\F^{X(k)}$, denoting $p=\Pr_{v}[ f|_{\T v}\not\in C_v]$, the following holds. If $p \leq \epsilon$ then
    \[\dist(f,C) \leq \rho(p).\]
\end{definition}

\begin{remark}
    If an HDX code is locally testable as per Definition \pref{def:LTC}, then the number of queries made by the tester is at most $\max_v|X_{+v}(k)|$ which is the maximal number of $k$-faces containing a vertex $v$. In a bounded-degree complex $X$ this is a constant number.
    Moreover, if each $C_v$ is an LDPC, we can reduce the query complexity further (without changing the code), as in the following Claim.
\end{remark}

\begin{claim} \label{claim:loctest-to-loctest}
    Suppose each local code $C_v$ is defined by at most $m_0$ parity checks, each looking at most $q_0$ bits. If an HDX code $\calC^k[X,\set{C_v}]$ is $(\epsilon,\rho(\cdot))$-locally testable per Definition \ref{def:LTC}, then it is $(q_0,\frac\epsilon{m_0}, \rho'(\cdot))$-locally testable as per Definition \ref{def:LTC-classic}, where $\rho'(x):=\rho(m_0 x)$.
\end{claim}
\begin{proof}
    For any vertex $v$,  $f|_{\T v}\in C_v$ iff none of the $m_0$ parity checks fail. By union bound, the fraction of vertices $v$ for which at least one of the $m_0$ parity checks fail is at most $m_0 \cdot \frac {\epsilon} {m_0} = \epsilon$. This means that  $\Pr_{v}[ f|_{\T v}\not\in C_v] \leq m_0p <\epsilon$, and by the local testability per Definition \ref{def:LTC} we deduce that $\dist(f,C) \leq \rho(m_0 p)$.
\end{proof}

\begin{definition}[Agreement testability of HDX Codes]\label{def:agrtest}
    Let $\rho:\Rnn\to \Rnn$ be a strictly increasing function, and let $\epsilon>0$. 
    The HDX code $\calC^k[X,\set{C_v}]$ is called $(\epsilon,\rho(\cdot))$-agreement testable if, for any given collection $\sett{z_v\in C_v}{v\in X(0)}$, if 
    \[\alpha := \Pr_{uv\in X(1)}[z_v(\T {uv}(k))\neq z_{u}(\T {uv}(k))] < \epsilon\] then there exists some $x\in C$ such that 
    \[ \Pr_{v\in X(0)}(z_v \neq x|_{\T v(k)}) \le \rho (\alpha) . \]
\end{definition}

\begin{claim}[Agreement testability implies local testability]\label{claim:agr}

    Let $X$ be a $k$-dimensional simplicial complex, and assume that every vertex is contained in the same number of $k$ faces. Let $C=\calC^k[X,\set{C_v}]$ be an HDX code. 
    If $C$ is $(\epsilon_0,\rho_0(\cdot))$-agreement testable then it is $(\epsilon_0/2,\rho_1(\cdot))$-locally testable, where $\rho_1(\xi):=\rho_0(2\xi)+\xi$.
\end{claim}
\begin{proof}
    Suppose $f\in\F_2^{X(k)}$.   
    Assume that $\varepsilon= \Pr_{v}[ f|_{\T v}\not\in C_v]$. Let $V^*=\sett{v\in X(0)}{f|_{\T v}\in C_v}$. Set $z_v= 
        \begin{cases}
            f|_{\T v}  & v\in V^* \\
            0 & v\not\in V^*
          \end{cases}
        $. 
    Choose a random edge $uv\in X(1)$. The probability that either $u$ or $v$ are not in $V^*$ is at most $2\varepsilon$. In the remaining probability they surely agree, so the disagreement is at most $\alpha \leq 2\varepsilon$. 
    If $2\varepsilon \leq \epsilon_0$, then by the assumption on agreement testability, there must be some codeword $h\in C$ such that 
    $\Pr_v[z_v\neq h|_{\T v}] \leq \rho_0(\alpha) \le \rho_0(2\varepsilon)$. 
    The codeword $h$ disagrees with $f$ on some $\T v$ either when $v\not\in V^*$ or when $z_v\neq h|_{\T v}$. This event is upper bounded by $\varepsilon + \rho_0(2\varepsilon)=\rho_1(\varepsilon)$. By assumption every vertex is contained in the same number of $k$-faces. So choosing a random vertex and then a random $k$-face containing it, is the same as choosing a random $k$-face. Therefore,
    \[\dist(f,C)\leq \dist(f,h) = \Pr_{s\in X(k)}[f(s)\neq h(s)]\leq 
    \Pr_{v\in X(0)}[v\not\in V^*] + \Pr_{v\in X(0)}[z_v\neq h|_{\T v}]\leq 
    \rho_1(\varepsilon).
    \]
\end{proof}

Since there has been some discussion of this in the literature \cite{firstK2022good}, we spell out how any HDX code that is agreement testable automatically gives rise to the ``local-view'' code that is two-query testable. The idea is to move from a codeword $f\in C$ to the collection $\set{z_v}_{v\in X(0)}$ of local views where $z_v = f|_{\T v}$.
\begin{claim}[Agreement testability implies $2$-query LTCs]\label{claim:2q}    
Suppose $C$ is an HDX code. Let $\Sigma$ be a finite alphabet such that $\card\Sigma = \max_v \card{C_v}$, and fix an injection $\sigma_v:C_v\to\Sigma$ for each $v\in X(0)$. Define
\begin{equation}\label{eq:LC}
    LC = \sett{z:X(0)\to \Sigma}{ \exists f\in C,\hbox{ s.t. } z(v) = \sigma_v(f|_{\T v})\,\forall v\in X(0)}.
\end{equation}
If $C$ is agreement testable, then $LC$ is locally testable with two queries. If $C$ is $(\epsilon, \rho(\cdot))$-agreement testable, then $LC$ is $(2, \epsilon, \rho'(\cdot))$-locally testable per Definition~\ref{def:LTC-classic}, where $\rho'(p) = p + \rho(p)$.
\end{claim}
\begin{proof}
    We only give a sketch. Given $z\in LC$, the tester will select a random edge $uv$ and read $z(u),z(v)\in \Sigma$. It will interpret $z(u),z(v)$ as words in $C_u,C_v$ respectively. This is done by computing $z_v = \sigma_v^{-1}(z(v))$ and similarly $z_u = \sigma_u^{-1}(z(u))$. This inversion may fail since $\sigma_v$ is an injection but not necessarily a bijection, in which case the tester rejects. Otherwise, the tester will accept iff $z_u|_{\T {uv}}=z_v|_{\T {uv}}$. The analysis follows from the definition of agreement testability: define $\hat{z}_v = \sigma^{-1}_v(z(v))$ for each $v \in X(0)$, where if the inversion fails we define $\hat{z}_v = 0 \in C_v$, then the probability $p \le \epsilon$ the tester fails is at least $\Pr_{uv \in X(0)} \left[ \hat{z}_u|_{X_{+uv}} \not= \hat{z}_v|_{X_{+uv}} \right] =: \alpha$, which by Definition~\ref{def:agrtest} means that there is some $x \in C$ such that $\Pr_{v \in X(0)} \left[ \hat{z}_v \not= x|_{X_{+v}(k)} \right] \le \rho(\alpha) \le \rho(p)$. Define $y : v \rightarrow \sigma_v(x|_{X_{+v}}) \in LC$, so that $\Pr_{v \in X(0)}\left[ \sigma_v(\hat{z}_v) \not= y(v) \right] \le \rho(p)$. 
    
    Now note that 
    \begin{align*}
        \dist(z, LC) 
        &= \min_{lc \in LC}\Pr_{v \in X(0)} \left[ z(v) \not= lc(v) \right] \\
        &\le \Pr_{v \in X(0)} \left[ z(v) \not= y(v) \right] \\
        &\le \Pr_{v \in X(0)} \left[ z(v) \not\in \sigma_v(C_v) \right] + \Pr_{v \in X(0)} \left[ \sigma_v(\hat{z}_v) \not= y(v) \right] \\
        &\le p + \rho(p),
    \end{align*}
    where in the last line we used that $\Pr_{v \in X(0)} \left[ z(v) \not\in \sigma_v(C_v) \right] \le \Pr_{uv \in X(1)} \left[ z(u) \not\in \Sigma_u(C_v) \vee z(v) \not\in \Sigma_v(C_v) \right] \le p$. 
\end{proof}

\subsection{Sheaves and HDX codes}\label{subsec:sheaves}
Meshulam shows in \cite{meshulam2018graph} how to view the expander codes of \cite{SipserS96} as a twisted homology of a graph with certain local coefficients. He also describes a higher dimensional generalization of systems of local coefficients attached to higher dimensional simlicial complexes. First and Kaufman \cite{firstK2022good} focus on the cohomological (as opposed to homological) variant and give a framework for studying codes as sheaves. The HDX codes we have defined can be placed in this framework, as we briefly explain next.

An $\F$-sheaf $\FX$ over a simplicial complex $X$ is a collection of $\F$-vector spaces $F(x)$ for every $x\in X$, together with linear maps $\res s t : F(s)\to F(t)$ for every pair of faces $s\subset t$. These maps are called {\em restriction maps}, and are required to satisfy certain transitive consistency, see more details in \cite{firstK2022good}.

A $k$-dimensional HDX code naturally gives rise to a sheaf as follows. 
\begin{definition}
    Let $\F$ be a field, let $X$ be a $k$-dimensional simplicial complex, and let $\set{C_s\subseteq \F^{\T s}}_{s\in X(k-1)}$ be a collection of $\F$-linear local codes. Define a sheaf over $X$ with respect to $\set{C_s}$ to be
\begin{itemize}
    \item $F(t)=\F$ for every $t\in X(k)$,
    \item $F(s) = C_s \subseteq \F^{\T s(k) }$ for every $s \in X(k-1)$,
    \item $F(r) = \sett{ f\in \F^{\T r(k)} }{ \forall s\in \T r(k-1),\; f|_{\T s(k)}\in C_s}$ for every $r \in X(i)$ and every $i$. In other words, $F(r)$ is isomorphic to the HDX code defined over $X_r$ with a local code appropriately isomorphic to $C_s$ at a face $s\setminus r$,
    \item Since $\T s \subset \T r$ whenever $s\supset r$, we can define the restriction maps by actual restriction.
\end{itemize}
\end{definition}
The coboundary operator from vertices to edges is $\delta:\oplus_{v\in X(0)} F(v)\to \oplus_{e\in X(1)}F(e)$ is defined by 
$\delta f (e) = \sum_{v\subset e} \res v e f(v)$ (and the boundary operator is defined as the dual).  
For full definitions of the coboundary and boundary operators, see \cite{firstK2022good} (the coboundary operators are defined in Section 4 and the boundary operators in Section 7.4).
\begin{claim}
    Let $\F$ be a field, let $X$ be a $k$-dimensional simplicial complex, and let $\set{C_s\subseteq \F^{\T s}}_{s\in X(k-1)}$ be a collection of $\F$-linear local codes. Let $C_s^\perp = \sett{f\in \F^{\T s}}{f\perp C_s}$ be the code dual to $C_s$, and let $\FX^\perp$ be the sheaf over $X$ with respect to $\set{C_s^\perp}$.
    Then, letting $\partial_k:C^k(X,\FX^\perp)\to C^{k-1}(X,\FX^\perp)$ denote the $k$-th boundary operator, $Z^k=Ker \partial_k$ is the HDX code $\calC^k[X,\set{C_s}]$. \qed
\end{claim}

Moreover, for the sheaf $\FX$ with respect to $\set{C_s}$, the cocycle code $Z_0 = Ker \delta_0$ is the related ``local-views'' code defined in \eqref{eq:LC}, given by replacing each codeword with its ensemble of local views (see the definition in \eqref{eq:LC} and the ensuing discussion).
Agreement testability of our code is equivalent to cosystolic expansion of the sheaf $\FX$ in dimension $0$, as proven in \cite[Proposition 7.6]{firstK2022good}.

Finally, First and Kaufman describe in \cite[Section 7.4]{firstK2022good} how a sheaf gives rise to a quantum CSS code. 

\section{Coset Complex and Code}\label{sec:complex}

\def\a{\alpha}
\subsection{The Triangle Complex}\label{subsec:X}
We describe a family of simplicial complexes $\set{X_n}_n$.
The construction is a variant of the coset complexes constructed by Kaufman and Oppenheim \cite{KaufmanO181}. 

Let \(\F=\F_q\) be a fixed finite field. Let $\varphi\in \F_q[t]$ be a primitive (and irreducible) polynomial of degree $n$ and let \inote{I guess we should denote $\F_{q^n}$ instead of $R_n$} \(R_n = \F_q [t] / \langle \varphi \rangle \cong \F_{q^n}\) (i.e. the ring of univariate polynomials of degree \(\leq n-1\) where multiplication is done modulo \(\varphi\)). Further assume that we choose $n$ so that $3\nmid q^n-1$ (this is easy as long as $q\not \equiv 1\mod 3$).

We define three matrix groups
\begin{equation}\label{eq:Kgroups}
\begin{aligned}
K_1 = \Set{ \begin{pmatrix}
1 & a t & c t^2\\
0 & 1 & b t \\
0 & 0 & 1
\end{pmatrix} \in M_3(R_n) },\\
K_2 = \Set{ \begin{pmatrix}
1 & 0 & 0 \\
ct^2 & 1 & a t \\
b t & 0 & 1
\end{pmatrix} \in M_3(R_n) },\\
K_3 = \Set{ \begin{pmatrix}
1 & a t& 0\\
0 & 1 & 0 \\
b t & c t^2 & 1
\end{pmatrix} \in M_3(R_n) }.
\end{aligned}
\end{equation}
and let $G=G_n$ be the group generated by $K_1,K_2,K_3$.
Clearly $G_n \subseteq SL_3(R_n)$. We will show that for $\phi$ and $n$ as we specified above, it will hold that $G_n = SL_3(R_n)$. 

We define three additional (smaller) subgroups, 
\begin{equation}\label{eq:Hgroups}
\begin{aligned}
H_1 = K_2\cap K_3 = \sett{ \begin{pmatrix}
1 & 0 & 0\\
0 & 1 & 0 \\
\a t & 0 & 1
\end{pmatrix}}{\a\in \F_q},\\
H_2 = K_1\cap K_3= \sett{ \begin{pmatrix}
1 & \a t& 0\\
0 & 1 & 0 \\
0 & 0 & 1
\end{pmatrix}}{\a\in \F_q},\\
H_3 = K_1\cap K_2 = 
\sett{ \begin{pmatrix}
1 & 0 & 0\\
0 & 1 & \a t \\
0 & 0 & 1
\end{pmatrix} }{\a\in \F_q}.
\end{aligned}
\end{equation}
\begin{claim}
    Each subgroup $H_{1}, H_{2}, H_{3}$ is isomorphic to the abelian group $(\F_q, +)$ via the isomorphism $\alpha \leftrightarrow h_i(\alpha)$ for 
    $h_i(\a)\in H_i$ the matrix with $\a t$ in the appropriate location. \qed\end{claim}The coset complex considered in this paper is 
\begin{equation}
    \label{eq:coset-complex}
    X = X[G; K_1,K_2,K_3]
\end{equation}
as per Definition \ref{def:cosetcomplex}. The group $G=G_n$ depends on the underlying ring $R=R_n$. When we let the size of the ring $R_n$ grow by increasing the degree $n$, we get a family of complexes $X_n$ as required. 

By definition $X$ is a $3$-partite simplicial complex satisfying the following, 
\begin{claim}\label{claim:coset-complex}
    Let $G,K_1,K_2,K_3,H_1,H_2,H_3$ be as above. Let $X = X[G;K_1,K_2,K_3]$. Then 
\begin{enumerate}
    \item The vertices correspond to cosets of $K_i$: $X(0) \cong G/K_1\;\sqcup\; G/K_2 \;\sqcup\; G/K_3$. Each vertex is contained in $q^3 = |K_i|$ triangles.      
    \item The edges of $X$ connect a vertex $u = g_iK_i$ to a vertex  $v = g_jK_j$ iff $i\neq j$ and the cosets intersect. In this case their intersection corresponds to a coset of $H_k$ for $k\neq i,j$. The elements of this coset are in $1-1$ correspondence with the set $T_{uv}$ of triangles containing the edge $uv$. There are exactly $|H_k|=q$ such triangles.
    \item $X(2)\cong G$. Moreover, given three vertices $u=g_1K_1$, $v=g_2K_2$, $w= g_3K_3$, the triangle $uvw$ belongs to $X(2)$ iff the three cosets have a nonempty intersection $g_1K_1\cap g_2K_2\cap g_3K_3 = \set g$.
    \item Assuming that $\varphi$ is primitive and $3 \nmid q^n-1$, then $G = \mathsf{SL}_3(R_n)$. 
\end{enumerate} 
\end{claim}
\begin{proof}
    The first three items are properties of coset complexes, as in \cite{KaufmanO181}. We prove the last item:

    Let $e_{ij}(\alpha)$ denote the matrix with $1$'s along the diagonal and $\alpha$ in the $(i,j)$'th position. (Note that $h_1(\alpha) = e_{12}(\alpha t), h_2(\alpha) = e_{23}(\alpha t)$, and $h_3(\alpha) = e_{31}(\alpha t)$.) The following so-called Steinberg relations hold for all $i \not= j \not= k$, 
    \[[e_{ij}(\alpha), e_{jk}(\beta)] = e_{ik}(\alpha \beta),\] where $[g,h] = ghg^{-1}h^{-1}$ denotes the commutator (this can be verified by direct calculation).

    We claim that using $h_1(1) = e_{12}(t)$, $h_2(1) = e_{23}(t)$, and $h_3(1) = e_{31}(t)$ it's possible to generate all matrices $e_{ij}(t^\beta)$ where $(i,j) \in \{ (1,2), (2,3), (3,1) \}$ and $\beta \equiv 1~\text{mod}~3$, as well as where $(i,j) \in \{ (1,3), (2,1), (3,2) \}$ and $\beta \equiv 2~\text{mod}~3$. We will prove this via induction. Suppose that for some $B \in \bbN_{\ge 0}$ all $e_{ij}(t^\beta)$ with $\beta \le B, \beta \equiv j-i~\text{mod}~3$ are generatable. Then, we will show the claim for $B+1$. If $B+1$ is a multiple of $3$, then there's nothing to show. Otherwise, if $B+1$ is $1~\text{mod}~3$, then by assumption we have that $e_{ik}(t^2)$ and $e_{kj}(t^{B-1})$ are both generatable where $k-i \equiv j-k \equiv 2 ~\text{mod}~3$, so $e_{ij}(t^{B+1}) = [e_{ik}(t^2), e_{kj}(t^{B-1})]$ is also generatable. If $B+1$ is $2~\text{mod}~3$, then we have that $e_{ik}(t)$ and $e_{kj}(t^B)$ are both generatable where $k-i \equiv j-k \equiv 1~\text{mod}~3$ by inductive hypothesis, so $e_{ij}(t^{B+1}) = [e_{ik}(t), e_{kj}(t^B)]$ are also both generatable.    

    Now, since $\varphi$ is primitive, it holds that the elements $t, \dots, t^{q^n-1}$ are all distinct and thus must cover all of $R_n \backslash \{ 0 \}$. If $3 \nmid q^n-1$, then the elements $t^{3\beta+1}$ where $0 \le \beta < q^n-1$ are also distinct and range over $R_n \backslash \{ 0 \}$, as do the elements $t^{3\beta+2}$. 
    The reason is that if $t$ generates the multiplicative group of $R_n=\F_q[t]/\iprod{\varphi}$, whose order is $q^n-1$, then as long as $3\nmid q^n-1$ then $t^3$ also generates the group as well.
    In other words, from $h_1(1),h_2(1),h_3(1)$ we can generate all $e_{ij}(\gamma)$ for any $i \not= j$ and $\gamma \in R_n \backslash \{ 0 \}$ (that is, we can generate all elementary matrices).

    To finish, it is well known that if we can generate all elementary matrices then we can generate all of $\textsf{SL}_3(R_n)$.
    
\end{proof}
This claim proves item (b) in \pref{thm:mainX}. Item (a) is proven in the next subsection.

\subsection{Structure of the Links}
\label{subsec:links}
In this section we describe the links. That is, we understand the structure of  \(X(K_i; H_{i+1}, H_{i-1})\). Without loss of generality we focus on the link $G_1 = (K_1/H_2\sqcup K_1/H_3, E_1)$.

Recall that the vertices of $G_1$ corresponds to cosets of the form $kH_2$ and $k'H_3$ while the edges correpsond to pairs of cosets $\{kH_2, k'H_3\}$ that have non-empty intersection $k_1H_1\cap k_2H_2$. For any coset $kH_2$ with representative $k= \begin{pmatrix}
        1 & 0 & ct^2 \\
        0 & 1 & bt \\
        0 & 0 & 1
    \end{pmatrix}$, we use $(*,b,c)$ to denote the vertice in $G_1$. Similarly for any coset  $k'H_3$ with representation $k'= \begin{pmatrix}
        1 & \alpha t & \gamma t^2 \\
        0 & 1 & 0 \\
        0 & 0 & 1
    \end{pmatrix}$ we use $(\alpha,*,\gamma)$ to denote the corresponding vertice in $G_1$. 
So an edge connects $(*,b,c)$ and $(\alpha,*,\gamma)$ iff the following equation has a solution in $\F^2_q$:
\begin{align}\label{eq:edge-in-G1}
    \begin{pmatrix}
        1 & xt & ct^2 \\
        0 & 1 & bt \\
        0 & 0 & 1
    \end{pmatrix} = \begin{pmatrix}
        1 & \alpha t & (\alpha y + \gamma)t^2 \\
        0 & 1 & y t \\
        0 & 0 & 1
    \end{pmatrix}.
\end{align} This system of equations is solvable iff $c = \alpha b + \gamma$. Therefore we can deduce the following statement about the degree of $G_1$.

\begin{claim}
    Every vertex in $G_1$ has degree $q$. 
\end{claim}

Furthermore let $A$ be the normalized adjacency matrix of $G_1$. By the edge characterization \cref{eq:edge-in-G1}, we derive the following spectral gap for $A$.

\begin{claim} \label{claim:eigenvalue}
    $|\lambda_2(A)| = \frac{1}{\sqrt{q}}$. 
\end{claim}

\begin{proof}

Let $B$ be the $|K_1/H_2|\times|K_1/H_3|$ unnormalized biadjacency matrix of $G_1$ such that 
\[A = \begin{pmatrix}
    \mathbf{0} & \frac{1}{q}\cdot B \\
    \frac{1}{q}\cdot B^T & \mathbf{0}
\end{pmatrix}.\]

Then the matrix $BB^T$ is the $|K_1/H_2|\times|K_1/H_2|$ matrix whose value in its $((*,b,c), (*,b',c'))$-th entry is the number of 2-step walks from $(*,b,c)$ to $(*,b',c'))$. Equivalently, by \cref{eq:edge-in-G1}, the value is also the number of solutions to the equations $c = x b + y$ and $c' = x b' + y$ over $\F^2_q$. Therefore

\[BB^T[(*,b,c), (*,b',c')] = 
\begin{cases}
1 \quad & \text{when } b \neq b' \\
q &\text{when } b = b' \land c = c'\\
0 &\text{o.w.}
\end{cases}\]

We can thus explicitly write $BB^T = (J_q - I_q)\otimes J_q + q \cdot I_q\otimes I_q$ where $J_q \in \R^{q\times q}$ is the all-ones matrix. 

So its top eigenvalue satisfies $\lambda_1(BB^T) = q^2$ and the second-largest eigenvalue is $\lambda_2(BB^T) = q$. From this we can deduced that $|\lambda_2(A)| =  \frac{1}{\sqrt{q}}$.
\end{proof}

\subsection{Embedding the Complex into a Vector Space}\label{subsec:embed}
Recall that there is a natural isomorphism between the group $G$ and the triangles of the complex, $X(2)$.  We describe a natural way to biject $G$ to a set of points $S\subset \F_q^{9n}$, \rnote{im changing all the $m$'s to $n$'s}
\[X(2) \cong G\stackrel \iota\hookrightarrow S\subset R^9\cong \F_q^{9n}.
\]
Every $g\in G$ is a $3\times 3$ matrix $(r_{ij})_{1\leq i,j\leq 3}$ such that $r_{ij}\in R$ for each $i,j$. An element in $R$ is a univariate polynomial of degree at most $n-1$ so it is specified by $n$ coefficients $r_{ij}(t) = \sum_{\ell=0}^{n-1} r_{ij}^{(\ell)} t^\ell$. We simply map each of the nine matrix entries into a vector of coefficients in $\F_q^n$ and concatenate them all: 
\begin{equation}\label{eq:embedding}
    ( r_{ij})\stackrel \iota\longmapsto \left(r^{(0)}_{11},r^{(1)}_{11},\ldots r^{(n-1)}_{11}, r^{(0)}_{12},r^{(1)}_{12},\ldots r^{(n-1)}_{12}, \ldots, r^{(0)}_{33},r^{(1)}_{33},\ldots r^{(n-1)}_{33}  \right)
\end{equation}
This embedding is clearly injective and it is linear in the coefficients of the matrix entries, namely $\iota(\alpha g+ \beta g') = \alpha\iota(g)+\beta\iota(g')$ for any $\alpha,\beta\in \F_q$ and $g,g'\in G$.
\begin{claim}\label{claim:line-embed}
    For $g\in G$ and $i=1,2,3$ let $\ell_{g,i}:\F\to gH_i$ be defined by $\ell_{g,i}(\a) = gh_i(a)$. Then $\iota\circ\ell_{g,i}:\F\to\F^n$ is an affine line that can be written as 
    \[\iota(\ell_{g,i}(\a)) = v_0 + \a v_i\] where $v_0=\iota(g)$ and for  $g = \begin{pmatrix} | & | & | \\ g_1 & g_2 & g_3 \\ | & | & | \end{pmatrix}$, we have
    \[
    v_1=\iota\left(\begin{pmatrix} | & 0 & 0 \\ tg_3 & 0 & 0 \\ | & 0 & 0 \end{pmatrix}\right);\quad v_2=\iota\left(\begin{pmatrix}  0 & |& 0 \\ 0 & tg_1  & 0 \\ 0& | & 0 \end{pmatrix}\right);\quad v_3=
    \iota\left(\begin{pmatrix} 0 & 0 & | \\ 0 & 0 & tg_2 \\ 0 & 0 & | \end{pmatrix}\right).
    \]
    Here the entries in $tg_i$ are taken modulo $\varphi(t)$, and $v_i\neq 0$ for $i=1,2,3$.

    Moreover, for any $g'\in gH_i$, the line $\iota\circ\ell_{g',i}$ is a re-parameterization of $\iota\circ\ell_{g,i}$, satisfying
    \[ \iota\circ\ell_{g',i}(\a) = \iota\circ\ell_{g,i}(\a+\a')\]
    for $\a'\in\F$ such that $g' = gh_i(\alpha')$.
\end{claim}
\begin{proof}
Fix $g\in G$. We prove the first part for $i=1$, the other cases are similar. The matrix $gh_1(\a)$ is obtained from $g$ by adding $\a t$ times the third column of $g$ to the first column of $g$, namely, 
\[gh_1(\a) = (g_1,g_2,g_3) + (\a t\cdot  g_3,0,0)\] 
Since the embedding $\iota$ is linear in the coefficients, we get that
\[\iota(gh_1(\a)) = \iota((g_1,g_2,g_3)) + \iota((\a t\cdot  g_3,0,0))= v_0 + \a v_1\] 
as in the claim. 

Regarding the moreover part, by definition $\ell_{g',i}(\a)= g'h_i(\a) = gh_i(\a')h_i(\a) = gh_i(\a+\a') = \ell_{g,i}(\a+\a')$. Since the expression is linear in $\a$, and since $\iota$ is additive, $\iota\circ\ell_{g,i}$ is clearly the same affine line as $\iota\circ\ell_{g',i}$, reparameterized.
\end{proof}
We define, for $i=1,2,3$, 
\begin{equation}\label{eq:lines}
        \L_n^i = \sett{\iota\circ\ell_{g,i}:\F\to\F^n}{g\in G},
\end{equation}
and $\L_n = \L_n^1\cup \L_n^2\cup \L_n^3$.
This establishes item (c) in \pref{thm:mainX}.
\subsection{The Global Code}\label{subsec:globalcode}
Let $RS(q,d)$ be the Reed Solomon code of degree $d$ over $\F_q$. Namely, $RS(q,d)$ is the set of length-$q$ tuples $(p(\a)\,:\,\a\in \F_q)$ where $p$ is a univariate polynomial of degree at most $d$.

For any three parameters $0\leq d_1,d_2,d_3\leq q$, we defined the code $C_{d_1,d_2,d_3}$ in two ways, see Definition \ref{def:code}.
First, we defined it as a (punctured) lifted Reed-Solomon code,
    \begin{align*}
    C^1 = C_{n,d_1,d_2,d_3} &= \sett{f:\bar X_n\to\F}{ \forall i=1,2,3,\,\lin\in \L_n^i,\; f \circ \lin \in RS(q,d_i)} \tag{\ref{eq:RMcode}}
    \end{align*}
where $\L_n^i$ is the set of affine lines defined in \eqref{eq:lines},
and then we defined it as an HDX code,
    \begin{align*}
        C^2 = C_{n,d_1,d_2,d_3} &= \sett{f:X_n(2)\to\F}{ \forall e\in X(1),\; f|_{\T e}\in C_e} \tag{\ref{eq:Tannercode} }
    \end{align*}
    where $C_e$ is isomorphic to $RS(q,d_i)$ for edges of type $i$. We now define $C_e$ more explicitly.  Recall from Claim \ref{claim:coset-complex} (item 2) that every edge $e$ of type $i$ corresponds to a coset $gH_i$, in the sense that the triangles containing $e$, which we denote $\T e$, correspond to the elements of $gH_i$. Choose, for each edge, one group element $g$ to be a coset representative. We write $\T e = \set{gh_i(\a)}_{\a\in\F}$ and define 
    \begin{equation}\label{eq:localcode-iso} C_e = \sett{f\in \F^{\T e}}{f(gh_i(\a))\hbox{ is a degree }d_i\hbox{ function of }\a}. \end{equation}
The definition of $C_e$ appears to depend on the choice of coset representative but it does not, because the degree of $f(gh_i(\a))$ as a function of $\a$ is the same as the degree of $f(gh_i(\a+\a')) = f(g'h_i(\a))$ as a function of $\a$.
\begin{claim}\label{claim:iso}
    The codes $C^1,C^2$ are isomorphic, and $\iota:G\to \bar X_n$ gives the isomorphism.
\end{claim}
\begin{proof}
    Fix $f\in C^1$. We define $\tilde f:G\to \F$ by $\tilde f = f \circ\iota$, and show $\tilde f\in C^2$.
    We need to check that $\tilde f|_{\T e}\in C_e$ for each $e$. 
    By definition this is true iff $(\tilde f(gh_i(\a)))_{\a} \in RS(q,d_i)$ for $e$ a type $i$ edge. But $\tilde f (gh_i(\cdot))= \tilde f\circ \ell_{g,i} = f\circ\iota\circ\ell_{g,i} \in RS(q,d_i)$ where the last inclusion is because $f\in C^1$ and $\iota\circ \ell_{g,i}\in \L_n^i$.

    The other direction is easy as well. Given $\tilde f\in C^2$, we let $f = \tilde f\circ \iota^{-1}:\bar X_n\to\F$  the unique function such that $f\circ\iota = \tilde f$. To check that $f\in C^1$ consider any line $\lin=\iota\circ \ell_{g,i}\in \L_n$ and observe that $f\circ\lin= f\circ(\iota\circ\ell_{g,i}) =\tilde f \circ\ell_{g,i}  \in RS(q,d_i)$. 
\end{proof}

Since there is a $1-1$ correspondence between $X(2)$ and $G$, we may also write codewords of $C^2$ as $f:G\to\F$. A group always acts transitively on itself by left multiplication. Moreover,
\begin{claim}\label{clm:trans}
    If $w:G\to\F_q$ is a codeword, then $w^g$ is a codeword, where $w^{g}(g') = w(gg')$
\end{claim}
\begin{proof}
    This is clear since for any $g'\in G$ and any $i=1,2,3$, $(w^g(g'h_i(\a))_\a = (w(gg'h_i(\a))_\a \in RS_q^{d_i}$.
\end{proof}

This establishes the transitivity of the code, as claimed in last item of \pref{thm:mainCode}. It implies that the restriction of our code to $K_i$ is isomorphic to the restriction of our code to any coset $gK_i$. To study the local view of the code at a link of a vertex it suffices to study its restriction to $K_i$ for $i=1,2,3$.

\subsection{Multiplication Property}\label{subsec:mult}
It is immediate that the codes have the multiplication property by inheritance from the Reed-Solomon code. Recall that for $w,w'\in \F^N$ we define $w''=w\odot w'\in \F^N$ by coordinate-wise product: $w''[i] = w[i]\cdot w'[i]$.
\begin{lemma}\label{lem:mult}
    Suppose $C = C_{n,d_1,d_2,d_3}$ and $C' = C_{n,d'_1,d'_2,d'_3}$, then for every $w \in C$ and $w'\in C'$, we have $w\odot w' \in C'' = C_{n,d_1+d'_1,d_2+d'_2,d_3+d'_3}$. 
\end{lemma}
\begin{proof}
For every $e\in X(1)$ the local views of $w|_{\T e},w'|_{\T e}$ are Reed-Solomon codewords of degrees $d_i,d'_i$ (relying on the definition in \eqref{eq:Tannercode}). So the coordinate-wise product, $w''|_{\T e}$, is a Reed-Solomon codeword of degree at most $d_i+d'_i$, as needed.
\end{proof}

\subsection{Distance} \label{subsec:distance}
The global code $C$ can also easily be shown to have constant relative distance,
\begin{lemma}[Distance]\label{lemma:distance}
If the relative distance of $C_e$ is at least $\delta>0$ for every $e\in X(1)$ then $C$ has relative distance at least $(\delta-2\gamma)(\delta-\gamma)\delta$. 
\end{lemma}
\begin{proof}
Let $0\neq x\in C$. 
Let $V^* = \sett{v\in X(0)}{x|_{\T v} \neq 0}$ and let $v\in V^*$. We will first show that at least $(\delta-\gamma)$ of the edges touching $v$ are non-zero. (An edge $e$ is non zero iff $x(e)\neq 0$.)

Let $A_v = \sett{u\in X_v(0)}{x|_{\T {uv}}\neq 0}$. Each  $u\in A_v$ has $0\neq x|_{\T {uv}}\in C_{uv}$ so $x|_{\T {uv}}$ must have at least $\delta$ fraction of non-zero entries. Each of these non-zero entries corresponds to some vertex $w$ such that $uvw\in \T{}$ with $x(uvw)\neq 0$ so $w\in A_v$. We found that for each $u\in A_v \subset X_v(0)$, at least $\delta$ of its neighbors (inside $X_v$) are in $A_v$. Since the graph $X_v$ is a $\gamma$-expander, the Alon-Chung lemma (Lemma \ref{lemma:AC}) implies that $|A_v| \geq (\delta - \gamma)\card{ X_v(0)}$. 

Observe that every $u\in A_v$ must itself be in $V^*$, so each $v\in V^*$ has at least $\delta-\gamma$ fraction of its neighbors in $V^*$. We can again apply Lemma \ref{lemma:AC}, (now using the fact that the graph $(X(0),X(1))$ is a $\gamma$-expander, to deduce $|V^*| \geq (\delta-2\gamma)|X(0)|$. We have seen that each $v\in V^*$ has at least $\delta-\gamma$ fraction of non-zero edges touching it, so the total fraction of non-zero edges is at least  $(\delta-2\gamma)(\delta-\gamma)$. Each such edge touches at least $\delta$ non-zero triangles, so the total fraction of non-zero triangles is at least $(\delta-2\gamma)(\delta-\gamma)\delta$ as claimed. 
\end{proof}
This along with the fact that $q$ is constant and $\gamma = 1/\sqrt{q}$ from Claim~\ref{claim:eigenvalue} establishes the second item in \pref{thm:mainCode}. 

We now prove a generalization of the above distance lemma to higher dimensional HDX codes

\begin{lemma}[Distance of $k$-Dimensional HDX Code] \label{lemma:hd-distance} 
    Let $X$ be a $k$-dimensional $\gamma$-one-sided local expander, and assume that for every $t \in X(k-1)$ we have a code $C_t \subset \{ f : X_{+t}(k) \rightarrow \bbF \}$ with minimum relative distance $\ge \delta > 0$. Define, for every $-1 \le i \le k-2$ and every face $t\in X(i)$, the code \[C_t = \sett{f:\T t(k)\to\F}{f|_{\T t}\in C_t\;\forall  t\in \T t(k-1)}.\]
    Then for every $-1\leq i<k-2$ and every $t\in X(i)$ the code $C_t$ has relative distance $\prod_{j = 0}^{k-1-i} (\delta - j\gamma)$. In particular, the code $C = C_\varnothing$ on the entire complex has relative distance $\prod_{j = 0}^k (\delta - j\gamma)$. 
\end{lemma}
Before giving the proof we remark that the distance in the lemma decays exponentially with the dimension, assuming that $\gamma \ll \delta$. This is necessary, as can be seen from the example the $(k+1)$-dimensional tensor code $C^{\otimes k+1}$, for $C$ any code with distance $\delta$. This code has distance $\delta^{k+1}$ and it can be viewed as an HDX code on the complete $k+1$-partite $k$-dimensional complex\footnote{There is a natural identification of $[n]^{k+1}$ with the faces of this complex. The link of every $k-1$ face is identified with a row in the appropriate direction.}. 
\begin{proof}
    Let $0 \not= x \in C_{\varnothing}$. For all $-1 \le i \le k-1$ and $s \in X(i)$ we define $A_s = \{ u \in X_s(0) ~|~ x|_{X_{+(s \cup \{ u \})}} \not= 0 \}$. We claim that if $x|_{X_{+s}} \not= 0$, then $|A_s| \ge \left( \delta - (k-1-i) \gamma \right) |X_s(0)|$, and furthermore that $C_s$ has relative distance $\prod_{j = 0}^{k-1-i} (\delta - j \gamma)$.

    To see this, we proceed by (downwards) induction on $i$. This is clearly the case for $i = k-1$. Now for $i < k-1$, let $t \in X(i)$. For $v \in X_t(0)$ such that $t \cup \{ v \} \in A_t$, we have that each $u \in A_{t \cup \{v\}}$ has $0 \not= x|_{X_{+(t \cup \{ u,v \})}} \in C_{t \cup \{ u,v \}}$ so $t \cup \{ u \} \in A_t$ as well. By the inductive hypothesis, we have that $|A_{t \cup \{ v \}}| \ge (\delta - (k-2-i)\gamma) |X_{t \cup \{ v\}}(0)|$ for all $v \in V^*_t$. Since $X_t$ is a $\gamma$-expander, the Alon-Chung lemma (Lemma~\ref{lemma:AC}) implies that $|A_t| \ge (\delta - (k-2-i)\gamma - \gamma) |X_t(0)| = (\delta - (k-1-i)\gamma) |X_t(0)|$. 

    Now, to compute the distance, we have that each $t \in A_v$ has a $\delta - (k-1-i)\gamma$ fraction of non-zero $(i+1)$-faces touching it, each of which has a $\prod_{j=0}^{k-2-i} (\delta - j\gamma)$ fraction of $k$-faces touching it, so the total fraction of non-zero $k$-faces in $X_{+t}$ is at least $\prod_{j=0}^{k-1-i} (\delta - j\gamma)$.
\end{proof}

\subsection{Local Code at a Vertex}
For each $v\in X(0)$, let
\begin{equation}
    C_v = \sett{f\in \F^{\T v}}{\forall e\ni v,\;f|_{\T e}\in C_e}.
\end{equation}
It is easy to see that our code can be written as  
\[ C =  \sett{f\in \F^{X(2)}}{\forall v,\;f|_{\T v}\in C_v}\] because we are simply aggregating the constraints differently than in \eqref{eq:Tannercode}.

\inote{ ...removed text from here.. }
\remove{Recall from Definition~\eqref{eq:Tannercode} $C^2 = \calC(X,\set{C_e})$ consists of all $f\in \F^{X(2)}$ such that $f|_{\T e}\in C_e$. 
We now define the localization of this code to a link of a vertex $v\in X(0)$. 
For a vertex $v\in X(0)$ let $\psi=\psi_v$ be a localization operator, $\psi f\in \F^{X_v(1)}$ is defined by $\psi f(uw) := f(uvw)$.  
Let
\begin{equation}
    C_v = \sett{f\in \F^{X_v(1)}}{\forall u\in X_v(0),\;f|_{E_u}\in \psi(C_{uv})}.
\end{equation} where $E_u$ are the edges of $X_v(1)$ that contain $u$.
Then
\begin{equation}
    C_v = \sett{f\in \F^{\T v}}{\forall e\ni v,\;f|_{\T e}\in C_e}.
\end{equation} }
What does $C_v$ look like when moving to $\bar X_n$?
\begin{lemma}
\label{lemma:loc-iso}
Fix $v\in X(0)$ a vertex of type $i$. \begin{itemize}
    \item $\iota(\T v)$ is a $3$ dimensional affine subspace in $\bar X_n$.
    \item The code $C_v\subset \F^{\T v}$ is isomorphic to $C_{d_x,d_y}$ for $d_x=d_{i+1}$, $d_y=d_{i-1}$, where
    we define $C_{d_x,d_y}$ by 
\begin{equation}\label{eq:def:vertexcode}
    C_{d_x,d_y} = \sett{f:\F_q^3\to \F_q}{ \forall a,b,c,\,\, deg_x(f(x,b,c))\leq d_x; \,\, deg_y (f(a,y,ay+c)) \leq d_y}. 
\end{equation} 
\end{itemize}
\end{lemma}
\begin{proof}
    Fix first $v = gK_1=K_1$ given by $g=id$. The elements of $\iota(K_1)$ are 
    \[\sett{\iota\left(\begin{pmatrix}
        1 & at & ct^2 \\
        0 & 1 & bt \\
        0 & 0 & 1
    \end{pmatrix}\right)}{a,b,c\in \F}\] and when we range over all possible choices of $a,b,c$ we get an $\F$-linear subspace.  If $v=gK_1$ for some $g\not\in K_1$, every element becomes
    \[\begin{pmatrix}
        | & | & | \\
        g_1 & g_2 & g_3 \\
        | & | & |
    \end{pmatrix}\cdot \begin{pmatrix}
        1 & at & ct^2 \\
        0 & 1 & bt \\
        0 & 0 & 1
    \end{pmatrix}
    = \begin{pmatrix}
        | & | & | \\
        g_1 & atg_1+g_2 & ct^2g_1+btg_2+g_3 \\
        | & | & |
    \end{pmatrix}\]
which after embedding into $\F^n$ becomes $\iota(g_1,g_2,g_3) + a\iota(0,tg_1,0) + b\iota(0,0,tg_2) + c\iota(0,0,t^2g_1)$. This is a $3$ dimensional affine subspace.
A similar proof applies to vertices of type $2,3$.

For the second item, we focus again on $v=K_1$. The code $C_v$ consists of all $f\in \F^{\T v}$ that, for each $e\ni v$, satisfy $f|_{\T e}\in C_e$. We identify functions on $\T v$ with functions on $\F^3$ through the isomorphism $\F^3\to K_1$ given by $(a,b,c) \mapsto \begin{pmatrix}
        1 & at & ct^2 \\
        0 & 1 & bt \\
        0 & 0 & 1
    \end{pmatrix}$. An edge $e\ni v$ corresponds to a coset of $H_2$ or $H_3$ in $K_1$, say $gH_2$, given by a coset representative $g= \begin{pmatrix}
        1 & 0 & ct^2 \\
        0 & 1 & bt \\
        0 & 0 & 1
    \end{pmatrix}$. The group elements of the coset are $gh_2(x)$ for all $x\in \F$,
    \begin{align*}
    \sett{\begin{pmatrix}
        1 & 0 & ct^2 \\
        0 & 1 & bt \\
        0 & 0 & 1
    \end{pmatrix} 
    \begin{pmatrix}
        1 & xt & 0 \\
        0 & 1 & 0 \\
        0 & 0 & 1
    \end{pmatrix}
    = 
    \begin{pmatrix}
        1 & xt & ct^2 \\
        0 & 1 & bt \\
        0 & 0 & 1
    \end{pmatrix}}{x\in \F_q},
    \end{align*}
    each corresponding to an triangle of $\T e$. So, by definition of $C_e$ (see \eqref{eq:localcode-iso}) the constraint $f|_{\T e}\in C_e$ translates to $f(gh_2(x))$ having degree $d_2$ in $x$. In other words, $f(x,b,c)$ must have degree at most $d_2$ in $x$.

Similarly, suppose the edge $e$ is $gH_3$ for $g= \begin{pmatrix}
    1 & at & ct^2 \\
    0 & 1 & 0 \\
    0 & 0 & 1
\end{pmatrix}$. The group elements of the coset are $gh_3(y)$ for all $y\in \F$, where 
\begin{align*}
\sett{    \begin{pmatrix}
        1 & at & ct^2 \\
        0 & 1 & 0 \\
        0 & 0 & 1
    \end{pmatrix} 
    \begin{pmatrix}
        1 & 0 & 0 \\
        0 & 1 & yt \\
        0 & 0 & 1
    \end{pmatrix}
    = 
    \begin{pmatrix}
        1 & at & (ay+c)t^2 \\
        0 & 1 & yt \\
        0 & 0 & 1
    \end{pmatrix}}{y\in \F_q}.
\end{align*}
The constraint $f|_{\T e}\in C_e$ translates to requiring that $f(a,y,ay+c)$ have degree is at most $d_3$ in $y$.

It is now clear that $C_v$ is isomorphic to $C_{d_2,d_3}$ when $v=K_1$. 
The same also holds for any $v=gK_1$ since by Claim \ref{clm:trans} the code is invariant under the action of $G$. This implies that $C_v \cong C_{v'}$, for any $v,v'$ of the same color (since the group action moves any $K_i$ coset to any other $K_i$ coset, it is thus transitive on each color class).

To complete the proof we check that for $v=K_2$ we have $C_v \cong C_{d_3,d_1}$ and for $v=K_3$  $C_v\cong C_{d_1,d_2}$.

Let us start with $v=K_2$. Our requirement is that $f:K_2\to\F_q$ evaluates to a degree $d_1$ polynomial on cosets of $H_1<K_2$ and a degree $d_3$ polynomial on cosets of $H_3<K_2$.
Fix some $g=\begin{pmatrix}
        1 & 0 & 0 \\
        ct^2 & 1 & at \\
        bt & 0 & 1
\end{pmatrix} \in K_2$. For any element $g_x = \begin{pmatrix}
        1 & 0 & 0 \\
        0 & 1 & xt \\
        0 & 0 & 1
\end{pmatrix} \in H_3$, and $g_y = \begin{pmatrix}
        1 & 0 & 0 \\
        0 & 1 & 0 \\
        yt & 0 & 1
\end{pmatrix} \in H_1$, we have 
\[ gg_x = \begin{pmatrix}
        1 & 0 & 0 \\
        ct^2 & 1 & (a+x)t \\
        bt & 0 & 1    
\end{pmatrix},\qquad gg_y = \begin{pmatrix}
        1 & 0 & 0 \\
        (c+ay)t^2 & 1 & at \\
        (b+y)t & 0 & 1    
\end{pmatrix}.\] 
Writing now $f(a,b,c) = f\left(\begin{pmatrix}
        1 & 0 & 0 \\
        ct^2 & 1 & at \\
        bt & 0 & 1
\end{pmatrix} \right)$, we require that for all $a,b,c$, \begin{itemize}
    \item $f(a+x,b,c)$ must have degree $d_3$ in $x$; and 
    \item $f(a,b+y,c+ay)$ must have degree at most $d_1$ in $y$.
\end{itemize}
This is clearly equivalent to requiring
\begin{itemize}
    \item $f(x,b,c)$ must have degree $d_3$ in $x$ for all $b,c$; and 
    \item $f(a,y,c+ay)$ must have degree at most $d_1$ in $y$ for all $a,c$.
\end{itemize} Namely, it is equivalent to requiring $f\in C_{d_3,d_1}$

Finally, for $v=K_3$, we  
fix some $g=\begin{pmatrix}
        1 & at & 0 \\
        0 & 1 & 0 \\
        bt & ct^2 & 1
\end{pmatrix} \in K_3$. For any element $g_x = \begin{pmatrix}
        1 & 0 & 0 \\
        0 & 1 & 0 \\
        xt & 0 & 1
\end{pmatrix} \in H_1$, and $g_y = \begin{pmatrix}
        1 & yt & 0 \\
        0 & 1 & 0 \\
        0 & 0 & 1
\end{pmatrix} \in H_2$, we have 
\[ gg_x = \begin{pmatrix}
        1 & at & 0 \\
        0 & 1 & 0 \\
        (b+x)t & ct^2 & 1    
\end{pmatrix},\qquad gg_y = \begin{pmatrix}
        1 & (y+a)t & 0 \\
        0 & 1 & 0 \\
        bt & (yb+c)t^2 & 1    
\end{pmatrix}.\] 
Writing now $f(a,b,c) = f\left(\begin{pmatrix}
        1 & at & 0 \\
        0 & 1 & 0 \\
        bt & ct^2 & 1
\end{pmatrix}\right)$, we require that for all $a,b,c$, \begin{itemize}
    \item $f(a,b+x,c)$ must have degree $d_1$ in $x$; and 
    \item $f(a+y,b,c+by)$ must have degree at most $d_2$ in $y$.
\end{itemize}
This is clearly equivalent to requiring that $f\in C_{d_1,d_2}$.

Moving to cosets $g K_i$, by the fact that the code is invariant under the group action, for any coset $gK_i$, the local code is isomorphic to $C_{d_{i+1},d_{i-1}}$ 
\end{proof}

\remove{
\begin{definition}[Local code at a vertex]\label{def:localcode}
We define a code $C_{q,d_x,d_y}\subseteq \set{ f:\F_q^3\to\F_q}$ by allowing a function $f(x,y,z)$  in the code if and only if
\begin{itemize}
    \item For every $b,c\in \F_q$, the function $f(x,b,c)$ has degree at most $d_x$ as a polynomial in $x$.
    \item For every $a,c\in \F_q$ the function $f(a,y,ay+c)$ has degree at most $d_y$ as a polynomial in $y$.
\end{itemize}
Formally, \[
C_{q,d_x,d_y} =\sett{f:\F_q^3\to\F_q}{\forall a,b,c\in \F_q,\quad deg_x(f(x,b,c))\leq d_x, \; deg_y(f(a,y,ay+c))\leq d_y}.
\]
\end{definition}}


\section{Rate}\label{sec:rate}
\subsection{Rate of Global Code} \label{sec:rate-global}
We analyze the rate of our code in two regimes. The first, is when the relative rate of the local codes $C_e$ is at least $2/3$. In this case a standard constraint counting implies that the global code has constant relative rate.
\begin{lemma}
    Suppose $\dim(C_e) > (2/3+\epsilon)q$ for each $e\in X(1)$. Then $\dim(C) > 3\epsilon$.
\end{lemma}
The second parameter regime is when the relative rate of the local codes $C_e$ is arbitrarily small. In this case we give a non-trivial lower bound by demonstrating a collection of linearly independent codewords. These are 
\[ C' = \sett{f|_S}{f:\F^{9m}\to\F\hbox { has degree }\leq d}\] where $d = \min(d_1,d_2,d_3)$ and $S\subset \F^{9m}$ is the image of $G$ when embedded into the vector space, see \eqref{eq:embedding}.
\begin{lemma}
    $\dim(C') \geq \dim(RM_d^{3m})$. If we choose $|\bbF|\approx poly(m)$ we get polynomial rate.
\end{lemma}
\begin{proof}
    Consider the upper triangular matrices with $1$ on the diagonal. This set of matrices belongs to $S$ and is isomorphic to $\F^{3m}$, so any polynomial in $9m$ variables that depends only on these $3m$ variables and has total degree at most $d$ gives rise to a distinct codeword in $C'$.
\end{proof}
An alternative way to bound the rate is as follows. If the ring $R$ is a field (which by Claim~\ref{claim:coset-complex} is true whenever $\varphi$ is a primitive polynomial), then the fraction of matrices with determinant $0$ is about $1/|R|$, which is tiny. Moreover, let
$S_r=\sett{m \in M_{3\times 3}(R)}{\det(m)=r}$.
Since every $0\neq r\in R$ has an inverse, there is a bijection between $S_{r_1}$ and $S_{r_2}$ for every $r_1,r_2\neq 0$, given by $m_1 \leftrightarrow m_2 = \begin{pmatrix}
r_2r_1^{-1} & 0 & 0\\
0 & 1 & 0 \\
0 & 0 & 1
\end{pmatrix} \cdot m_1$. In particular, our set $S=S_1$ bijects to each $S_r$.  
Each $S_r$ for $r\neq 0$ supports a copy of our complex, obtained by a linear transformation of the entire vector space which moves shifts the lines.
We leave open the question of obtaining better bounds on the global rate. 

\subsection{Rate of Local Code at a Vertex} \label{sec:rate-local}
In this section we analyse the rate of the code $C_{q,d_x,d_y}\subseteq \set{ f:\F_q^3\to\F_q}$ defined in \eqref{eq:def:vertexcode}. In this section, we will only consider the case where $q$ is a prime $p$.

Since $f(x, b, c)$ lies on a degree $d_x$ polynomial in $x$ for any $b, c$, we can write $f(x, y, z)$ as a polynomial that is degree $d_x$ in $x$ and degree $p-1$ in $y$ and $z$:
\[
    f(x, y, z) = \sum_{\substack{0 \le i \le d_x \\ 0 \le j, k \le p-1}} c_{ijk} x^i y^j z^k.
\]

Plugging in $(x, y, xy + z)$, we get that 
\[
    f(x, y, xy + z) = \sum_{\substack{0 \le i \le d_x \\ 0 \le j, k \le p-1}} c_{ijk} x^i y^j (xy + z)^k.
\]
We can reduce $f(x, y, xy + z)$ modulo $y^p - y$ to get a polynomial $g(x, y, z)$ that is degree $\le p-1$ in $y$ and $z$, and $\le d_x+p-1$ in $x$. We let $g_j(x,z)$ and $g_{jk}(x)$ denote the coefficient of $y^j$ and $y^jz^k$ respectively, so that 
\[
    g(x, y, z) = \sum_{0 \le j \le p-1} g_j(x, z) y^j = \sum_{0 \le j, k \le p-1} g_{jk}(x) y^j z^k.
\]
Note that we can also write 
\begin{align*}
    g_j(x, z) &= \begin{cases}
        h_j(x, z) & j = 0 \\
        h_j(x, z) + h_{j+p-1}(x, z) & j \not= 0
    \end{cases} \\
    g_{jk}(x) &= \begin{cases}
        h_{j,k}(x) & j = 0 \\
        h_{j,k}(x) + h_{j+p-1,k}(x) & j \not= 0
    \end{cases}. 
\end{align*}

The condition that $f(a, y, ay+c)$ lies on a degree $d_y$ polynomial in $y$ for all $a$ and $c$ means that for all $d_y < j \le p-1$, it holds that $g_j(x, z) = 0$ for all $x, z$. In particular, for any $0 \le k \le p-1$, it must hold that $g_{jk}(x) = 0$ for all $x$. So, if we know that $g_{jk}(x)$ is degree $\le p-1$ in $x$, then in fact all coefficients in $g_{jk}(x)$ must be $0$. We will use this fact extensively in the analysis of the rate.

First, we will need the following lemma.

\begin{lemma} \cite{GesselV85} \rnote{todo: check if GesselV handles the finite field version; right now im only convinced its ok if youre working in field of prime order} \label{lemma:invertible_matrix}
    Assuming that $r \le k \le m < p$, the following matrix has full rank in $\bbF_p$:
    \[
        \begin{pmatrix}
            \binom{m}{k} & \binom{m}{k-1} & \cdots & \binom{m}{k-r} \\
            \binom{m-1}{k} & \binom{m-1}{k-1} & \cdots & \binom{m-1}{k-r} \\
            \vdots & \vdots &  & \vdots \\
            \binom{m-r}{k} & \binom{m-r}{k-1} & \cdots & \binom{m-r}{k-r}
        \end{pmatrix}
    \]
\end{lemma}

\begin{proof}
    We divide the entries of row $i$ by $(m-i)!$ (where the top row is row $0$), and multiply the entries of column $j$ by $(m-k+j)!(k-j)!$ (where the leftmost column is column $0$). Since $p > m$ and $m \ge k \ge r \ge i,j$, both $(m-i)!$ and $(m-k+j)!(k-j)!$ are nonzero in $\bbF_p$. We obtain that the rank of the above matrix is equivalent to the rank of the below matrix:
    \[
        \begin{pmatrix}
            1 & 1 & 1 & \cdots & 1 \\
            h_1(m-k) & h_1(m-k+1) & h_1(m-k+2) & \cdots & h_1(m-k+r) \\
            h_2(m-k) & h_2(m-k+1) & h_2(m-k+2) & \cdots & h_2(m-k+r) \\
            \vdots & \vdots & \vdots &  & \vdots \\
            h_r(m-k) & h_r(m-k+1) & h_r(m-k+2) & \cdots & h_r(m-k+r)
        \end{pmatrix},
    \]
    where $h_i(x) = x(x-1)(x-2)\cdots(x-i+1)$ is a degree $i$ polynomial.

    This matrix is invertible. To see this, let $\alpha_j = m-k-j$. The above matrix is rewritten as 
    \[
        \begin{pmatrix}
            1 & 1 & 1 & \cdots & 1 \\
            h_1(\alpha_0) & h_1(\alpha_1) & h_1(\alpha_2) & \cdots & h_1(\alpha_r) \\
        h_2(\alpha_0) & h_2(\alpha_1) & h_2(\alpha_2) & \cdots & h_2(\alpha_r) \\
            \vdots & \vdots & \vdots &  & \vdots \\
            h_r(\alpha_0) & h_r(\alpha_1) & h_r(\alpha_2) & \cdots & h_r(\alpha_r)
        \end{pmatrix}.
    \]
    Using the smaller degree monomials in the rows above it, the polynomial $h_i$ in each row can iteratively be made into simply the monomial $x^i$. In particular, the rank of the above matrix is the same as the rank of the Vandermonde matrix, which has full rank.
\end{proof}

\begin{lemma} \label{lem:dx+dy}
    Assuming that $p \ge d_x + d_y + 2$, then for $j + k > d_x + d_y$ and for all $0 \le i \le d_x$, it holds that $c_{ijk} = 0$. In other words, the combined degree of $y$ and $z$ in $f(x,y,z)$ is at most $d_x + d_y$.
\end{lemma}

\begin{proof}
    Recall that we've written 
    \[
        f(x,y,z) = \sum_{0 \le i \le d_x} c_{ijk} x^iy^jz^k. 
    \]
    We will backwards induct on the value of $j + k$ to show that $c_{ijk} = 0~\forall i \in [0, d_x]$. The base case is $j + k = 2p-1$. In this case, there are no terms with $j + k = 2p-1$, so $c_{ijk} = 0$ for all $i \in [0, d_x]$.

    Now, assume that for $s > d_x + d_y$ and $j + k > s$, it holds that $c_{ijk} = 0~\forall i \le d_x$. We will prove that for all $j+k = s$, it holds that $c_{ijk} = 0$ for $i \in [0, d_x]$.

    To begin, we express the polynomial $g_{jk}(x)$ in terms of the nonzero coefficients $c_{ijk}$.

    \begin{claim} \label{claim:coeffs-in-gjk}
        Assume that $c_{ij'k'} = 0$ for $j' + k' > s$ and $i \in [0, d_x]$. Then for $j + k = s$, it holds that
        \[
            g_{jk}(x) = \sum_{\substack{0 \le i \le d_x \\ k \le k' \le p-1}} c_{i,s-k',k'} \cdot \binom{k'}{k} x^{k'-k+i}. 
        \]
    \end{claim}

    \begin{proof}
        Consider a term $c_{i'j'k'}x^{i'}y^{j'}(xy+z)^{k'}$ that appears in $f(x,y,xy+z)$. This expands to
        \[
            c_{i'j'k'} \cdot \sum_{\substack{0 \le i \le d_x \\ 0 \le k \le k'}} \binom{k'}{k} x^{k'-k+i}y^{j'+k'-k}z^{k}.
        \]
        Thus, the coefficient $c_{i'j'k'}$ appears in the expression of $g_{jk}(x)$ when $(j' + k' - k \modstar{p}) = j$. Note that all the sum of the exponents of $y$ and $z$ in every term of the above expression is $(j'+k'-k) + k = j'+k'$. Then, the coefficients that appear in the expression of $g_{jk}(x)$ must satisfy the property that $j'+k' \in \{ s, s + p-1 \}$. Since we've assumed that $c_{i'j'k'} = 0$ for $j' + k' > s$, it holds that the $c_{i'j'k'}$ that appear are precisely those where $j' + k' = s$. We thus obtain the formula in the claim.
    \end{proof}

    
        

    We now continue to prove the inductive step for $j + k = s$. We split into three cases.
    
    \vspace{0.5em}
    \begin{caseof}
    
    \hspace{-3em}
    \case{$s > p-1 + d_y$.}{
        In this case, we proceed by backwards induction on the value of $k$ to show that $c_{ijk} = 0$. For $k > p-1$, there are no such terms, so $c_{i,s-k,k} = 0$ for all $i \in [0, d_x]$.

        Now assume the inductive hypothesis that for all $k' > k$ it holds that $c_{i,s-k',k'} = 0$. Using Claim~\ref{claim:coeffs-in-gjk}, we have that 
        \begin{align*}
            g_{s-k,k}(x) 
            &= \sum_{\substack{0 \le i \le d_x \\ k \le k' \le p-1}} c_{i,s-k',k'} \cdot \binom{k'}{k} x^{k'-k+i} \\
            &= \sum_{0 \le i \le d_x} c_{i, s-k, k} \cdot x^i,
        \end{align*}
        where in the second line we've used the inductive hypothesis.
        But $s - k \ge s - (p-1) > d_y$, so $g_{s-k, k}(x) = 0~\forall x \in \bbF_p$. Since $g_{s-k,k}(x)$ is a polynomial of degree $d_x < p$, this implies that all the coefficients $c_{i, s-k, k}$ are equal to $0$. This completes the inductive step.
    }
    \vspace{0.5em}
    \case{$p-1 < s \le p-1 + d_y$.}{
        In this case, to prove that $c_{i,s-k,k} = 0$ for all $i \in [0, d_x]$ and $k \in [0,p-1]$, our strategy is to find independent linear relations between these coefficients. It will be convenient to view a matrix which represents the data given by Claim~\ref{claim:coeffs-in-gjk}, as follows:
        \[
        \kbordermatrix{
            & g_{d_y+1,s-d_y-1} & g_{d_y+2,s-d_y-2} & g_{d_y+3,s-d_y-3} & \cdots & g_{p-1,s-p+1} \\
            c_{i, s-p+1, p-1} & \binom{p-1}{s-d_y-1} x^{p-s+d_y+i} & \binom{p-1}{s-d_y-2} x^{p-s+d_y+i+1} & \binom{p-1}{s-d_y-3} x^{p-s+d_y+i+2} & \cdots & \binom{p-1}{s-p+1} x^{2p-s+i-2} \\
            c_{i, s-p+2, p-2} & \binom{p-2}{s-d_y-1} x^{p-s+d_y+i-1} & \binom{p-2}{s-d_y-2} x^{p-s+d_y+i} & \binom{p-2}{s-d_y-3} x^{p-s+d_y+i+1} & \cdots & \binom{p-2}{s-p+1} x^{2p-s+i-3} \\
            c_{i, s-p+3, p-3} & \binom{p-3}{s-d_y-1} x^{p-s+d_y+i-2} & \binom{p-3}{s-d_y-2} x^{p-s+d_y+i-1} & \binom{p-3}{s-d_y-3} x^{p-s+d_y+i} & \cdots & \binom{p-3}{s-p+1} x^{2p-s+i-4} \\
            \vdots & \vdots & \vdots & \vdots &  & \vdots \\
            c_{i,d_y+1,s-d_y-1} & \binom{s-d_y-1}{s-d_y-1} x^i & \binom{s-d_y-1}{s-d_y-2} x^{i+1} & \binom{s-d_y-1}{s-d_y-3} x^{i+2} & \cdots & \binom{s-d_y-1}{s-p+1} x^{p-d_y+i-2} \\
            c_{i,d_y+2,s-d_y-2} & 0 & \binom{s-d_y-2}{s-d_y-2} x^i & \binom{s-d_y-2}{s-d_y-3} x^{i+1} & \cdots & \binom{s-d_y-2}{s-p+1} x^{p-d_y+i-3} \\
            c_{i,d_y+3,s-d_y-3} & 0 & 0 & \binom{s-d_y-3}{s-d_y-3} & \cdots & \binom{s-d_y-3}{s-p+1} x^{p-d_y+i-4} \\
            \vdots & \vdots & \vdots & \vdots &  & \vdots \\
            c_{i,p-1,s-p+1} & 0 & 0 & 0 & \cdots & \binom{s-p-1}{s-p-1} x^i
        }
        \]
        The rows correspond to the coefficients $c_{i,s-k,k}$ for $k = p-1, \dots, s-p+1$ as labeled on the left. Each row actually corresponds to $d_x+1$ coefficients, for $i \in [0, d_x]$, but for sake of space we condense these into a single row. 
        
        The columns correspond to $g_{d',s-d'}(x)$ for $d' \in [d_y+1, p-1]$. Because $d' > d_y$, these $g_{d',s-d'}(x)$ all should be $0$. The way to read off the value of $g_{d',s-d'}(x)$ is by looking at the corresponding column, and summing the product of each entry with the corresponding row coefficient $c_{i,s-k',k'}$, remembering that each row actually corresponds to $d_x+1$ rows for $i \in [0, d_x]$. The entries of the matrix were chosen according to the expansion of $g_{d',s-d'}(x)$ given in Claim~\ref{claim:coeffs-in-gjk}.


        Now, to prove that $c_{i,s-k,k} = 0$ for all $i \in [0, d_x]$ and $k \in [0,p-1]$, we will backwards induct on the value of $i+k$, which we will denote by $t$. So assume that for $i' + k' > t$ that $c_{i',s-k',k'} = 0$. This is true for $t = d_x + (p-1)$, since there do not exist coefficients with larger values of $i' + k'$ (as $i' \le d_x$ and $k' \le p-1$).

        Consider all the coefficients with $i + k = t$, meaning coefficients $c_{i, s-t+i, t-i}$ where $i \in [0, d_x]$ and $t-i, s-t+i \in [0, p-1]$. These coefficients correspond to $\le d_x+1$ consecutive rows of the matrix with increasing values of $i$. Note that in each column, these coefficients only contribute to a single monomial $x^{i'}$: namely, in the column for $g_{d', s-d'}$, the exponents of $x$ that appear for $i+k = t$ are all equal to $i' = i + (k - (s-d')) = t - s + d'$. Thus, for any $g_{d', s-d'}(x)$ which has degree $< p$, we obtain a linear constraint on $c_{i, s-t+i, t-i}$ by restricting the corresponding column corresponding to the rows corresponding to the coefficients (we may also drop the powers of $x$, so that the coefficients of the linear constraint are simply binomial coefficients). For instance, for $t = p$, the column $g_{d', s-d'}$ gives us the constraint 
        \[
            \begin{pmatrix}
                c_{1, s-p+1, p-1} &
                c_{2, s-p+2, p-2} &
                \cdots &
                c_{d_x, s-p+d_x, p-d_x} 
            \end{pmatrix}
            \begin{pmatrix}
                \binom{p-1}{s-d'} \\
                \binom{p-2}{s-d'} \\
                \vdots \\
                \binom{p - d_x}{s - d'}
            \end{pmatrix}
            = 0.
        \]
        

        Our strategy may therefore be summarized as follows: for each $t \le p-1 + d_x$, assuming that there are $m$ coefficients of the form $c_{i,s-t+i,t-i}$, we will find $m$ consecutive columns such that the degree of $x$ in $g_{d',s-d'}(x)$ is $< p$. These columns will also satisfy that if we restrict the matrix to these columns and to the rows corresponding to the coefficients, the resulting $m \times m$ matrix has diagonal that lies above the lower left triangle of $0$'s. Then, by Lemma~\ref{lemma:invertible_matrix}, we know these $m$ linear constraints are independent, telling us that $c_{i, s-t+i, t-i} = 0$ for each $i$.

        The columns will be chosen as follows.
        \begin{itemize}
        
        \item 
            For $t \ge p-1$, we are considering the $m = p-t+d_x$ coefficients 
            \[
                c_{t-p+1, s-p+1, p-1}, ~c_{t-p+2, s-p+2, p-2}, ~\dots,~ c_{d_x, s-t+d_x, t-d_x}.
            \]
            We pick the first $p-t+d_x$ columns of the matrix, corresponding to $g_{d_y+1,s-d_y-1},~ g_{d_y+2,s-d_y-2}, ~\dots, ~g_{p+d_x+d_y-t,s+t-p-d_x-d_y}$. Each of $g_{d',s-d'}(x)$ for $d' \in [d_y+1, p+d_x+d_y-t]$ is a polynomial in $x$. Using the inductive hypothesis that $c_{i',s-k',k'} = 0$ for $i' + k' > t$, we see that the maximum degree of $x$ in any of these polynomials is $\le i + k - \min (s-d') = t - (s+t-p-d_x-d_y) = p+d_x+d_y-s < d_x+d_y < p$. 

        \item 
            For $s-d_y-1 \le t < p-1$, we are looking at the $d_x+1$ coefficients 
            \[
                c_{0,s-t,t}, c_{1,s-t+1,t-1}, \dots, c_{d_x,s-t+d_x,t-d_x}.
            \]
            Consider the $d_x+1$ columns corresponding to $g_{d_y+1, s-d_y-1}, ~\dots, ~g_{d_x+d_y+1, s-d_x-d_y-1}$. Because $t \ge s-d_y-1$, the submatrix has nonzero diagonal. The largest exponent of $x$ in any of these columns is $\le i + k - (s-d_x-d_y-1) \le t - s + d_x+d_y - 1 < d_x+d_y-1 < p$. 

        \item 
            For $t < s - d_y - 1$, let $\hat{d} = \min(d_x, p-1-s+t)$. We consider the $\hat{d} + 1$ coefficients
            \[
                c_{0,s-t,t}, ~c_{1,s-t+1,t-1}, ~\dots, ~c_{\hat{d},p-1,t-\hat{d}}.
            \]
            Look at the $\hat{d} + 1$ columns corresponding to $g_{s-t,t},~\dots,~g_{s-t+\hat{d},t-\hat{d}}$. This matrix has nonzero diagonal. The largest degree of $x$ in any of these columns is $\le i+k - (t-\hat{d}) \le \hat{d} < p$. 
            
        \end{itemize}

    }
    \vspace{0.5em}
    \case{$d_x+d_y < s \le p-1$.}{
        Similar to the previous case, we consider the following matrix, interpreted the same way as before.
        \begin{align*}
            \kbordermatrix{
            & g_{d_y+1,s-d_y-1} & g_{d_y+2,s-d_y-2} & g_{d_y+3,s-d_y-3} & \cdots & g_{s,0} \\
            c_{i,0,s} & \binom{s}{s-d_y-1} x^{i+d_y+1} & \binom{s}{s-d_y-2} x^{i+d_y+2} & \binom{s}{s-d_y-3} x^{i+d_y+3} & \cdots & \binom{s}{0} x^{i+s} \\
            c_{i,1,s-1} & \binom{s-1}{s-d_y-1} x^{i+d_y} & \binom{s-1}{s-d_y-2} x^{i+d_y+1} & \binom{s}{s-d_y-3} x^{i+d_y+2} & \cdots & \binom{s-1}{0} x^{i+s-1} \\
            c_{i,2,s-2} & \binom{s-2}{s-d_y-1} x^{i+d_y-1} & \binom{s-2}{s-d_y-2} x^{i+d_y} & \binom{s-2}{s-d_y-3} x^{i+d_y-1} & \cdots & \binom{s-2}{0} x^{i+s-2} \\
            \vdots & \vdots & \vdots &  & \vdots \\
            c_{i,d_y+1,s-d_y-1} & \binom{s-d_y-1}{s-d_y-1} x^{i} & \binom{s-d_y-1}{s-d_y-2} x^{i+1} & \binom{s-d_y-1}{s-d_y-3} x^{i+2} & \cdots & \binom{s-d_y-1}{0} x^{i+s-d_y-1} \\
            c_{i,d_y+2,s-d_y-2} & 0 & \binom{s-d_y-2}{s-d_y-2} x^i & \binom{s-d_y-2}{s-d_y-3} x^{i+1} & \cdots & \binom{s-d_y-2}{0} x^{i+s-d_y-2} \\
            c_{i,d_y+3,s-d_y-3} & 0 & 0 & \binom{s-d_y-3}{s-d_y-3} x^i & \cdots & \binom{s-d_y-3}{0} x^{i+s-d_y-3} \\
            \vdots & \vdots & \vdots &  & \vdots \\
            c_{i,s,0} & 0 & 0 & 0 & \cdots & \binom{0}{0} x^i
            }.
        \end{align*}
        As in the previous case, we will downwards induct on the value of $i + k = t$. Assume that for $i' + k' > t$ it holds that $c_{i',s-k',k'} = 0$. This is true for $t = d_x + s$ (the maximal possible value of $i+k$) since there do not exist coefficients with larger values of $i' + k'$.

        We again consider all $m$ coefficients satisfying $i + k = t$. These occupy a number of consecutive rows of the matrix. Furthermore, for any column with degree in $x$ less than $p$, restricting that column to those rows gives a linear constraint on the coefficients. Thus, we again demonstrate, for each $t$, $m$ consecutive columns that each have exponent in $x$ less than $p$. Restricting these columns to the coefficient rows will result in a $m \times m$ submatrix lying above the lower left triangle of $0$'s, which is full rank by Lemma~\ref{lemma:invertible_matrix}, implying that all $m$ coefficients must in fact be $0$.

        \begin{itemize}
        \item 
            For $t \ge s$, we are interested in the $s-t+d_x+1$ coefficients 
            \[
                c_{t-s,0,s}, ~c_{t-s+1,1,s-1}, ~\dots, ~c_{d_x,s-t+d_x,t-d_x}.
            \]
            Look at the first $s-t+d_x+1$ columns, which correspond to $g_{d', s-d'}$ for $d' \in [d_y+1, s-t+d_x+d_y+1]$. The maximum degree of any of these polynomials is $i + k - (t-d_x-d_y-1) = d_x + d_y +1 < p$.

        \item 
            For $s-d_y-1 \le t < s$, we consider the $d_x+1$ coefficients
            \[
                c_{0,s-t,t}, ~ c_{1,s-t+1,t-1}, ~ \dots, ~ c_{d_x,s-t+d_x,t-d_x}.
            \]
            The columns we are interested in are $g_{d_y+1,s-d_y-1}, ~ \dots, ~g_{d_x+d_y+1,s-d_x-d_y-1}$. Since $t \ge s-d_y-1$, the submatrix has nonzero diagonal. The maximum degree of any of these polynomials is $\le i + k - (s-d_x-d_y-1) = t - s + d_x + d_y + 1 < d_x + d_y + 1 < p$.

        \item 
            For $t < s-d_y-1$, let $\hat{d} = \min(d_x, t)$. Then we are considering the following $\hat{d} + 1$ coefficients.
            \[
                c_{0,s-t,t}, ~ c_{1,s-t+1,t-1}, ~ \dots, ~ c_{\hat{d}, s-t+\hat{d}, t-\hat{d}}.
            \]
            For these, we will look at the $\hat{d}+1$ columns corresponding to $g_{s-t,t}, ~ \dots, ~ g_{s-t+\hat{d},t-\hat{d}}$. This clearly has nonzero diagonal, and the maximum degree of $x$ in any of these polynomials is $i + k - (t-\hat{d}) = \hat{d} < p$.
        \end{itemize}
        
    }
    \end{caseof}
    
\end{proof}

\begin{theorem} \label{lemma:local-rate}
    The dimension of the local code is $\frac12 \cdot (d_x+1)(d_y+1)(d_x+d_y+2)$.
\end{theorem}

\begin{proof}
    From Lemma~\ref{lem:dx+dy}, we know that $c_{ijk} = 0$ for all $i \in [0, d_x]$ and $j+k > d_x + d_y$. Thus we can write 
    \begin{align*}
        f(x, y, z) &= \sum_{\substack{0 \le i \le d_x \\ 0 \le j + k \le d_x + d_y}} c_{ijk} x^i y^j z^k.
    \end{align*}
    Note that $c_{i'j'k'}x^{i'}y^{j'}(xy+z)^{k'} = \sum_{0 \le k \le k'} \binom{k'}{k} x^{k'-k+i'}y^{k'-k+j'}z^{k}$ and in particular the sum of the exponents of $y$ and $z$ is always $j' + k'$, so in particular the value of $g_{jk}(x)$ only depends on coefficients $c_{i'j'k'}$ were $j' + k' = j + k$ (note also that we're now in the setting where $j + k \le d_x + d_y < p$, so $g(x, y, z) = f(x, y, xy + z)$ without having to reduce modulo $y^p - y$).

    Again, we will use the fact that for all $j > d_y$ and $k \in [0, p-1]$, it holds that $g_{jk}(x) = 0~\forall x \in \bbF_p$. Note that $g_{jk}(x)$ could have degree as large as $2d_x + d_y$ which could be larger than $p$. We thus let $\hat{g}_{jk}(x)$ denote $g_{jk}(x)$ reduced modulo $x^p - x$. The condition that $g_{jk}(x) = 0~\forall x \in \bbF_p$ then is equivalent to $\hat{g}_{jk}(x) \equiv 0$. 

    We will work through the polynomials $\hat{g}_{jk}(x)$ and set each to $0$ by choosing coefficients $c_{ijk}$ appropriately. We will consider all the polynomials $\hat{g}_{jk}$ with a fixed value of $j+k$, denoted by $s$, simultaneously.

    As such, let $d_y < s \le d_x + d_y$. The relevant polynomials $g_{jk}(x)$ that evaluate to $0$ everywhere are $g_{d_y+1, s-d_y-1}, \dots, g_{s, 0}$, and the relevant coefficients are $c_{i, 0, s}, \dots, c_{i, s, 0}$, where $i \in [0, d_x]$. The dependency of $g_{d', s-d'}$ on the coefficients is given in the following matrix, interpretted the same as before.

    \begin{align*}
        \kbordermatrix{
        & g_{d_y+1,s-d_y-1} & g_{d_y+2,s-d_y-2} & g_{d_y+3,s-d_y-3} & \cdots & g_{s,0} \\
        c_{i,0,s} & \binom{s}{s-d_y-1} x^{i+d_y+1} & \binom{s}{s-d_y-2} x^{i+d_y+2} & \binom{s}{s-d_y-3} x^{i+d_y+3} & \cdots & \binom{s}{0} x^{i+s} \\
        c_{i,1,s-1} & \binom{s-1}{s-d_y-1} x^{i+d_y} & \binom{s-1}{s-d_y-2} x^{i+d_y+1} & \binom{s}{s-d_y-3} x^{i+d_y+2} & \cdots & \binom{s-1}{0} x^{i+s-1} \\
        c_{i,2,s-2} & \binom{s-2}{s-d_y-1} x^{i+d_y-1} & \binom{s-2}{s-d_y-2} x^{i+d_y} & \binom{s-2}{s-d_y-3} x^{i+d_y-1} & \cdots & \binom{s-2}{0} x^{i+s-2} \\
        \vdots & \vdots & \vdots &  & \vdots \\
        c_{i,d_y+1,s-d_y-1} & \binom{s-d_y-1}{s-d_y-1} x^{i} & \binom{s-d_y-1}{s-d_y-2} x^{i+1} & \binom{s-d_y-1}{s-d_y-3} x^{i+2} & \cdots & \binom{s-d_y-1}{0} x^{i+s-d_y-1} \\
        c_{i,d_y+2,s-d_y-2} & 0 & \binom{s-d_y-2}{s-d_y-2} x^i & \binom{s-d_y-2}{s-d_y-3} x^{i+1} & \cdots & \binom{s-d_y-2}{0} x^{i+s-d_y-2} \\
        c_{i,d_y+3,s-d_y-3} & 0 & 0 & \binom{s-d_y-3}{s-d_y-3} x^i & \cdots & \binom{s-d_y-3}{0} x^{i+s-d_y-3} \\
        \vdots & \vdots & \vdots &  & \vdots \\
        c_{i,s,0} & 0 & 0 & 0 & \cdots & \binom{0}{0} x^i
        }.
    \end{align*}

    We will set the coefficients $c_{i, s-k, k}$ starting with those with the largest value of $t := i + k$, which is $d_x + s$, and proceeding downwards. We will show that if we've already set all coefficients $c_{i,s-k,k}$ with $i + k > t$, then the degree of $\hat{g}_{d', s-d'}(x)$ is $\le d'+t-s$ (that is, all larger monomials have been set to $0$ by the previous choices of coefficient assignments). This is certainly satisfied in the base case: if $t = d_x + s$, then since $i \le d_x$ we have that the largest power of $x$ in any $g_{d', s-d'}(x)$ is $d_x+d' = d' + t - s$. This largest exponent can only decrease when we pass to $\hat{g}_{d', s-d'}(x)$. 
    
    Now, suppose that we're part way through this process and have already set the values for all $c_{i, s-k, k}$ where $i + k > t$. The number of ways we have to set the values $c_{t-k, s-k, k}$ is as follows.
    
    \begin{itemize}

    \item 
        For $t = d_x+s, ~d_x+s-1, ~\dots, ~d_x+d_y+1$, the coefficients of interest are $c_{t-s, 0, s}, ~c_{t-s+1,1,s-1}, ~\dots, ~c_{d_x,s+d_x-t,t-d_x}$. We will show that these must all be set to $0$. We look at the columns corresponding to $g_{d_y+1, s-d_y-1}, ~\dots, ~g_{d_x+d_y-t+s+1, t-d_x-d_y-1}$. Note that there are as many columns as coefficients. By the inductive assumption, the maximal degree of column $g_{d', s-d'}$ is $\le d' + t - s \le d_x + d_y + 1 < p$. Thus the coefficient of $x^{d'+t-s}$ in $g_{d', s-d'}(x)$, which is equal to $\sum_{t-d_x \le k \le s} \binom{k}{s-d'} c_{t-k,s-k,k}$, must also be $0$. This gives us the following linear system of equations:
        \[
            \begin{pmatrix}
                c_{t-s,0,s} \\
                c_{t-s+1,1,s-1} \\
                \vdots \\
                c_{d_x,s-t+d_x,t-d_x} 
            \end{pmatrix}^T 
            \cdot 
            \begin{pmatrix}
                \binom{s}{s-d_y-1} & \binom{s}{s-d_y-2} & \cdots & \binom{s}{t-d_x-d_y-1} \\
                \binom{s-1}{s-d_y-1} & \binom{s-1}{s-d_y-2} & \cdots & \binom{s-1}{t-d_x-d_y-1} \\
                \vdots & \vdots & & \vdots \\
                \binom{t-d_x}{s-d_y-1} & \binom{t-d_x}{s-d_y-2} & \cdots & \binom{t-d_x}{t-d_x-d_y-1} 
            \end{pmatrix}
            = 0,
        \]
        where the matrix is also a top left submatrix of the aforementioned matrix, ignoring the powers of $x$. Since $s > s-d_y-1$, this matrix has nonzero diagonal and by Lemma~\ref{lemma:invertible_matrix} is invertible. This implies that $c_{t-k,s-k,k} = 0$ for all $t-d_x \le k \le s$. 

        Now, we've set all $c_{ijk}$ with $j+k\ge t$. It remains to show that $\hat{g}_{d',s-d'}(x)$ is now degree $\le d'+t-s-1$. We already have that $\hat{g}_{d',s-d'}(x)$ was degree $\le d'+t-s$ from the inductive assumption. For any $\hat{g}_{d',s-d'}(x)$, the coefficient of the $x^{d'+t-s}$ term is the sum of the coefficients of $x^{d'+t-s}$ and $x^{d'+t-s+(p-1)}$ in $g_{d',s-d'}(x)$. But recall that we've set all $c_{ijk}$ with $j + k \ge t$ to $0$, so both these coefficients in $g_{d', s-d'}$ must be $0$. Therefore, $\hat{g}_{d', s-d'}(x)$ is degree $\le d' + t - s - 1$. 

    \item 
        Next, for $t = d_x + d_y, \dots, s+1$, we are looking at the $s + d_x - t + 1 > s-d_y$ coefficients $c_{t-s, 0, s}, \dots, c_{d_x, s-t+d_x, t-d_x}$. We have so far that all $\hat{g}_{d', s-d'}(x)$ have degree $\le d'+t-s$. The coefficient of $x^{d'+t-s}$ in $\hat{g}_{d',s-d'}(x)$ is equal to the sum of the coefficients of $x^{d'+t-s}$ and $x^{d'+t-s+(p-1)}$ in $g_{d', s-d'}(x)$, which in turn is equal to
        \[
            \sum_{t-d_x \le k \le s} \binom{k}{s-d'} \cdot c_{t-k, s-k, k} + \sum_{0 \le k \le s} \binom{k}{s-d'} \cdot c_{t-k+(p-1),s-k,k}, \numberthis \label{eqn:hatg1}
        \]
        where in the second summation the terms $c_{t-k+(p-1),s-k,k}$ that don't exist are understood to be $0$. Note that we've already set the values of $c_{t-k+(p-1),s-k,k}$. Thus, in order to set~\eqref{eqn:hatg1} to $0$, we need to choose $c_{t-k,s-k,k}$, $t-d_x \le k \le s$ so that 
        \[
            \sum_{t-d_x \le k \le s} \binom{k}{s-d'} \cdot c_{t-k, s-k, k} = -\sum_{0 \le k \le s} \binom{k}{s-d'} \cdot c_{t-k+(p-1),s-k,k}.
        \]
        There are $s-d_y$ such equations for the $s-d_y$ polynomials $\hat{g}_{d', s-d'}(x)$, which are all independent by Lemma~\ref{lemma:invertible_matrix} since $s \ge s-d_y-1$. Then, there are $p^{s+d_x-t+1-(s-d_y)} = p^{d_x+d_y+1-t}$ ways to choose the coefficients $c_{t-k,s-k,k}$. We remark also that once $c_{t-s,0,s},\dots,c_{d_x,s-t+d_x,t-d_x}$ are set, all $\hat{g}_{d',s-d'}(x)$ must have degree $\le d'+t-s-1$ since we chose the values so that the coefficient of $x^{d'+t-s}$ was $0$ for all columns.
        
    \item 
        The next case is $t = s, ~s-1, ~\dots, ~d_x$. In this case, we are looking at the $d_x + 1$ coefficients $c_{0,s-t,t}, \dots, c_{d_x, s-t+d_x, t-d_x}$. We have so far that all $\hat{g}_{d',s-d'}(x)$ have degree $\le d' + t - s$, and the coefficient of $x^{d' + t - s}$ in $\hat{g}_{d',s-d'}(x)$ is equal to the sum of the coefficients of $x^{d'+t-s}$ and $x^{d'+t-s+(p-1)}$ in $g_{d',s-d'}(x)$. Similar to the previous case, this results in $s-d_y$ independent linear equations (by Lemma~\ref{lemma:invertible_matrix}, using the fact that $t \ge s-d_y-1$. Thus, there are $p^{d_x+d_y+1-s}$ ways to set $c_{0,s-t,t}, \dots, c_{d_x, s-t+d_x, t-d_x}$. Note that after we've set these coefficients, the degree of $\hat{g}_{d',s-d'}(x)$ necessarily must be $\le d'+t-s-1$ by choice of these coefficients.

    \item 
        If $t = d_x-1, \dots, s-d_y$, then the coefficients we care about are $c_{0,s-t,t}, \dots, c_{t,s,0}$. We have that $\hat{g}_{d',s-d'}(x)$ has degree $\le d'+t-s$ and wish to set the $c_{t-k,s-k,k}$ so that the coefficient of $x^{d'+t-s}$ is $0$. As before, this results in $s-d_y$ linearly independent equations, which are independent because $t \ge s-d_y-1$. Thus, there are $p^{t+1-s+d_y}$ ways to set the coefficients $c_{t-k,s-k,k}$. The degrees of all $s-d_y$ polynomials $\hat{g}_{t-k,s-k,k}(x)$ are now $\le d' + t - s - 1$.

    \item 
        For $t = s-d_y-1, \dots, 0$, we are considering the $t+1$ coefficients $c_{0,s-t,t}, \dots, c_{t,s,0}$. Note also that for $d' < s-t$, we must already have that $\hat{g}_{d',s-d'}(x) \equiv 0$ since the maximum degree, if it exists, is already $\le d'+t-s$. So, we look at the $t+1$ polynomials $\hat{g}_{s-t,t}(x), \dots, \hat{g}_{s,0}(x)$. Since we want to set the coefficient of $x^{d'+t-s}$ to be $0$, this results in $t+1$ linearly independent equations in $c_{0,s-t,t},\dots, c_{t,s,0}$. So there is exactly one way to set $c_{0,s-t,t}, \dots, c_{t,s,0}$ to make all these coefficients $0$.
        
    \end{itemize}

    In total, for $d_y < s \le d_x + d_y$, if $C_s$ is the number of ways to assign all the coefficients $c_{i,s-k,k}$, then 
    \begin{align*}
        \log_p C_s 
        &= 0 + \sum_{t=s+1}^{d_x+d_y} (d_x+d_y+1-t) + \sum_{t=d_x}^s (d_x+d_y+1-s) + \sum_{t=s-d_y}^{d_x-1} (t+1-s+d_y) + 0 \\
        &= (d_y+1)(d_x+d_y+1-s).
    \end{align*}
    Furthermore, if $j+k = s \le d_y$, then $f(x,y,xy+z)$ is always degree $\le d_y$ in $y$, so all such coefficients $c_{ijk}$ are permissible. There are $(d_x+1)\cdot \binom{d_y+2}{2}$ such coefficients. In total, this gives that the dimension of the local code is 
    \begin{align*}
        (d_x+1) \cdot \binom{d_y+2}{2} + \sum_{s=d_y+1}^{d_x+d_y} \left( (d_y+1) (d_x+d_y+1-s) \right) 
        = \frac12 \cdot (d_x+1)(d_y+1)(d_x+d_y+2). 
    \end{align*}

\end{proof}

\section{Code Testability}
In this section, we will work with the field size $q = p$ being a prime. The main reason is that we will need results from Section~\ref{sec:rate-local} about the degree of codewords in $C_{d_x, d_y}$, viewed as low degree polynomials. \inote{explained degree of a codeword} \rnote{added}

\subsection{Testability of the Local Code at a Vertex}

In this section we prove that $C_v \cong C_{d_x,d_y}$ (as defined in Lemma~\ref{lemma:loc-iso}) is agreement-testable whenever $d_x + d_y < p/2$. \rnote{changed, i think this is the right condition?} Our definition of agreement testability (see Definition \ref{def:agrtest}) applies to HDX codes. Let us restate it in a form that is specialized for $C_v$:
\begin{definition}  
The local code $C_v$ defined in \eqref{eq:def:vertexcode} is $(\epsilon,\rho(\cdot))$-agreement testable if whenever we are given a collection of $z_e\in C_e$ for each $e\ni v$, such that 
    \[\a := \Pr_{uw\in X_v(1)}[z_{uv}(T_{uvw})\neq z_{wv}(T_{uvw})] < \epsilon\] then there exists some $x\in C_v$ such that 
\[ \Pr_{u\in X(0)}(z_{uv} \neq x|_{T_{uv}}) \le \rho (\alpha) . \]
\end{definition} 
Recall $C_v \cong C_{d_x,d_y}$. The local views have two kinds. Let us see the case where $v=K_1$: Cosets of $H_2$ correspond to lines $(x,b,c)$ for all $b,c\in \F$. Cosets of $H_3$ correspond to lines $(a,y,ay+c)$ for all $a,c\in \F$.
Therefore, the local views can be packaged through $X,Y$, which will be collections of degree $d_2$ and $d_3$ polynomials on the lines corresponding to cosets of $H_2$ and $H_3$, respectively.

The following theorem will immediately imply local testability. 
\begin{theorem}
\label{thm:local-testability}
    Let $X,Y: \F^3\to\F$ be such that
\begin{itemize}
    \item For each  $b,c$, the $x$ degree of $X(x,b,c)$ is at most $d_x$. 
    \item For each $a,c$, the $y$ degree of $Y(a,y,ay+c)$ is at most $d_y$. 
\end{itemize}
Then if $\bbP[X(x,y,z) \not= Y(x,y,z)] = \delta^3$ so that $p \ge 2(d_x + d_y) + 5\delta p$, there exists codeword $Q(x,y,z) \in C_{d_x,d_y}$ such that
\[\bbP_{b,c}[X(x,b,c) \not= Q(x,b,c)] + \bbP_{a,c}[Y(a,y,ay+c) \not= Q(a,y,ay+c)] \le 4\delta.\]
\end{theorem}
Before proving Theorem~\ref{thm:local-testability}, let us see that it immediately implies agreement testability. 

\begin{corollary} \label{cor:agr-test-local}
    Assuming that $d_x + d_y < \frac{p}{2}$, the local code $C_v = C_{d_x, d_y}$ is $\left( \left( \frac{p - 2(d_x + d_y)}{5p} \right)^3, 4 (\cdot)^{1/3} \right)$-agreement testable. \rnote{changed so holds for general $d_x + d_y$ up to $p/2$}
\end{corollary}

\begin{proof}
    Without loss of generality, let $v$ be the coset $K_1$ in $X(G; K_1,K_2,K_3)$ (the argument applies to other choices of $v$ since $G$ acts transitively on the code $C_v$, see  \cref{clm:trans}). 
    Then every edge $e\ni v$ is of type $2$ or type $3$. Recall from Theorem~\ref{thm:mainX} that type 2 edges $e$ are cosets $\{gh_2(x)\}_x$ while type 3 edges $e$ correpsond to cosets of the form $\{gh_3(y)\}_y$.
    Given a collection of local views $z_e\in C_e$, we define functions $X,Y: \F^3 \to \F$ in the following way:
    \begin{align*}
        &X(x,b,c) = z_{gh_2}(gh_2(x)) ~\text{where } g = \iota^{-1}(0,b,c), \\
        &Y(a,y,ay+c) = z_{gh_3}(gh_3(y)) ~\text{where } g = \iota^{-1}(a,0,c).
    \end{align*}
    Here $\iota$ is the embedding of elements in $K_1$ to $\F^3$ as used in the proof of \cref{lemma:loc-iso}. Since $z_{gh_2}(gh_2(x))$ has degree at most $d_2$ and $z_{gh_3}(gh_3(y))$ has degree at most $d_3$, $X$ and $Y$ satisfy the conditions in the theorem statement for $d_x = d_2, d_y =d_3$. Also by construction of $X$ and $Y$
    \begin{align*}
        \Pr_{uw\in X_v(1)}[z_{uv}(T_{uvw})\neq z_{wv}(T_{uvw})] = \bbP_{x,y,z}[X(x,y,z) \not= Y(x,y,z)]
    \end{align*}

    Similarly, given a codeword $Q\in C_{d_x.d_y}$, by \cref{lemma:loc-iso} we can define its corresponding codeword $f \in C_v$ as 
    \[f(g) = Q(\iota(g)).\] 
    Then the disagreement probability between $z_e$ and $f$ satisfies:
    \[ 2\Pr_{u\in X(0)}(z_{uv} \neq f|_{T_{uv}}) = \bbP_{b,c}[X(x,b,c) \not= Q(x,b,c)] + \bbP_{a,c}[Y(a,y,ay+c) \not= Q(a,y,ay+c)].\]
    
    Now applying \cref{thm:local-testability} to $z_e$ and $f$ we have that the local code $C_v$ is $\left( \left( \frac{p - 2(d_x + d_y)}{5p} \right)^3, 4 (\cdot)^{1/3} \right)$-agreement testable.

\end{proof}


We now move to prove Theorem~\ref{thm:local-testability}.
Let $S$ denote the set of all points $(x,y,z)$ on which $X(x,y,z) \not= Y(x,y,z)$. Consider all polynomials $e(x,y,z)$ that are degree $\delta p$ in $x$, and degree $\delta p$ in $y$ in $e(x,y,xy+z)$. By \cref{lemma:local-rate} , the dimension of the space of such polynomials is $(\delta p + 1)^3$. Thus, by dimension counting, there is a nonzero polynomial $E(x,y,z)$ that evaluates to $0$ on all points of $S$. Then, we have that for all $x,y,z$, $X(x,y,z) E(x,y,z) = Y(x,y,z) E(x,y,z)$.

We have that $X(x,y,z) E(x,y,z)$ is degree $d_x + \delta p$ in $x$, and $Y(x,y,xy+z) E(x,y,xy+z)$ is degree $d_y + \delta p$ in $y$. Thus there is a polynomial $P(x,y,z)$ that is degree $d_x + \delta p$ in $x$, and degree $d_y + \delta p$ in $y$ in $P(x,y,xy+z)$, that agrees on all points. That is, 
\[
    X(x,y,z) E(x,y,z) = Y(x,y,z) E(x,y,z) = P(x,y,z) ~\forall x,y,z.
\]
We'd like to formally divide $P(x,y,z)$ by $E(x,y,z)$ to obtain a polynomial $Q(x,y,z)$ that is degree $d_x$ in $x$ and degree $d_y$ in the skew-$y$ direction. 


\begin{lemma}
    Let $E(x,y,z)$ be a polynomial of degree $(a,b)$ in the $x$ and skew-$y$ directions, and let $P(x,y,z)$ be a polynomial of degree $(a+d_x, b+d_y)$ in the $x$ and skew-$y$ directions. If there exists $Y_x, Z_x \subseteq \bbF_p$ such that $E(x,y_0,z_0)$ divides $P(x,y_0,z_0)$ for $(y_0, z_0) \in Y_x \times Z_x$, and $X_y, Z_y \subseteq \bbF_p$ such that $E(x_0, y, x_0 y + z_0)$ divides $P(x_0, y, x_0 y + z_0)$ for $(x_0, z_0) \in X_y \times Z_y$, and if 
    \[
        \min(|Y_x|, |Z_x|) \ge 2(d_x + d_y + 2b) + a ~~~\text{and}~~~ \min(|X_y|, |Z_y|) \ge 2(d_x + d_y + 2a) + b,
    \] 
    then $E(x,y,z)$ divides $P(x,y,z)$.
\end{lemma}

\begin{proof}
    Assume without loss of generality that $a \ge b$ (otherwise we can consider the change of basis $P'(x,y,z) = P(y,x,xy-z)$ and $E'(x,y,z) = E(y,x,xy-z)$). We can also assume that $P$ and $E$ share no common factors, as follows: Let $F(x,y,z)$ be the largest common factor of $P(x,y,z)$ and $E(x,y,z)$. Assume by way of contradiction that $F \not\equiv E$ and that $F(x,y,z)$ has degree $(e,f)$ in the $x$ and skew-$y$ directions. Set
    \[
        P(x,y,z) \equiv \bar{P}(x,y,z) F(x,y,z) ~\text{and}~ E(x,y,z) \equiv \bar{E}(x,y,z) F(x,y,z). 
    \]
    We now divide $P$ and $E$ by $F$ and apply the lemma to $\bar{P}$ and $\bar{E}$. Let $\bar{Y}_x, \bar{Z}_x$ be such that $\bar{E}(x,y_0,z_0) | \bar{P}(x,y_0,z_0)$ for all $(y_0, z_0) \in \bar{Y}_x \times \bar{Z}_x$ and let $\bar{X}_y, \bar{Z}_y$ be such that $\bar{E}(x_0, y, x_0y+z_0) | \bar{P}(x_0,y,x_0y+z_0)$ for all $(x_0, z_0) \in \bar{X}_y \times \bar{Z}_y$. Then note that for any $y_0 \in Y_x$ either $\bar{E}(x,y_0,z_0) | \bar{P}(x,y_0,z_0)$ for all $z_0 \in Z_x$ or $F(x, y_0, z) = 0$ for all $x \in \bbF$ and $z \in Z_x$. If $\Sigma_y$ denotes all the values $y_0$ for which the latter occurs, since $|Z_x| > e + f$ which is at least the degree of $z$ in $F(x, y, z)$ by Lemma~\ref{lem:dx+dy}, we can write $F(x, y, z) = \prod_{y_0 \in \Sigma_y}(y - y_0) \cdot F'(x,y,z)$, from which it follows that $|\Sigma_y| \le f$. This implies that $|\bar{Y}_x| \ge |Y_x| - f$. Similarly we can show that $|\bar{Z}_x| \ge |Z_x| - f$, $|\bar{X}_y| \ge |X_y| - e$, and $|\bar{Z}_y| \ge |Z_y| - e$. The conditions of the lemma are satisfied because $\min(|\bar{Y}_x|, |\bar{Z}_x|) \ge 2(d_x + d_y + 2b) + a - f \ge 2(d_x + d_y + 2(b-f)) + (a-e)$ and $\min(|\bar{X}_y|, |\bar{Z}_y|) \ge 2(d_x + d_y + 2a) + b - e \ge 2(d_x + d_y + 2(a-e)) + (b - f)$. \rnote{changed the lemma assumption and this entire paragraph -- i think there was an error earlier but hopefully its ok now}
    

    So, we can assume that $P(x,y,z)$ and $E(x,y,z)$ have no common factors. We will use this assumption to obtain a contradiction. Write 
    \begin{align*}
        P(x,y,z) &\equiv P_0(y,z) + P_1(y,z) x + \cdots + P_{a+d_x}(y,z) x^{a+d_x} \\
        E(x,y,z) &\equiv E_0(y,z) + E_1(y,z) x + \cdots + E_a(y,z) x^a.
    \end{align*}
    Then, $P(x,y,z)$ and $E(x,y,z)$ have a common factor if there exists $A(x,y,z)$ that is degree $\le a-1$ in $x$ and $B(x,y,z)$ that is degree $\le a+d_x-1$ in $x$ such that 
    \[
        P(x,y,z) A(x,y,z) = E(x,y,z) B(x,y,z),
    \]
    or 
    \begin{align*}
        P_{a+d_x} A_{a-1} &= E_a B_{a+d_x-1} \\
        P_{a+d_x-1} A_{a-1} + P_{a+d_x} A_{a-2} &= E_{a-1} B_{a+d_x-1} + E_a B_{a+d_x-2} \\
        &\vdots \\
        P_0 A_0 &= E_0 B_0.
    \end{align*}
    These equations are linear in $A_i$ and $-B_i$ and can be summarized by the following $(2a + d_x) \times (2a + d_x)$ matrix: 
    \[
        M(P,E)(y,z) = 
        \begin{pmatrix}
            P_{a+d_x} & P_{a+d_x-1} & \cdots & \cdots & \cdots & P_0 & 0 & \cdots & 0 \\
            0 & P_{a+d_x} &  & \cdots & \cdots & P_1 & P_0 & \cdots & 0  \\
            \vdots & \ddots & \ddots & & & & \ddots & \ddots & \vdots \\
            0 & \cdots & 0 & P_{a+d_x} & \cdots & \cdots & \cdots & P_1 & P_0 \\
            E_a & E_{a-1} & \cdots & E_1 & E_0 & 0 & \cdots & \cdots & 0 \\
            \vdots & \ddots & & & & & \ddots & & \vdots \\
            \vdots & & \ddots & & & & & \ddots & 0 \\
            0 & \cdots & \cdots & \cdots & 0 & E_a & E_{a-1} & \cdots & E_0 
        \end{pmatrix}
    \]
    consisting of $a$ rows with $P$ entries and $a+d_x$ rows of $E$ entries.
    
    We define $R(P,E)(y,z)$, the resultant of $P$ and $E$, to be the polynomial in the coefficients of $P$ and $E$ (viewing both as a polynomial in $x$) obtained by taking the determinant of $M(P,E)$. Solutions $A$ and $B$ exist if $R(P,E) = 0$, which would give our contradiction. 

    Viewing $R(P,E)$ as a polynomial in $y$ and $z$, and recalling by Lemma~\ref{lem:dx+dy} that $P$ (resp. $E$) is total degree $\le 2(d_y+b)$ (resp. $\le 2b$) in $y$ and $z$, we see that $R(P,E)$ has degree $\le a \cdot 2(d_y+b) + (a+d_x) \cdot 2b$.

    Meanwhile, for each $(y_0, z_0) \in Y_x \times Z_x$, we have that $E(x, y_0, z_0) | P(x, y_0, z_0)$, so each of the top $a$ rows are linear combinations of the bottom $a+d_x$ rows. Thus, the polynomial $R(P,E)$ has a zero at each $(y_0, z_0) \in Y_x \times Z_x$ of multiplicity $a$. Because we assumed that 
    \[
        a \cdot \min (|Y_x|, |Z_x|) 
        \ge a \cdot \left( 2(d_x + d_y + 2b) + a \right)
        > a \cdot 2(d_y+b) + (a+d_x) \cdot 2b ,
    \]
    $R(P,E)(y,z)$ must be the zero polynomial by Schwarz-Zippel. This implies that $P(x,y,z)$ and $E(x,y,z)$ must have non-trivial common factor when considered as polynomials in $x$, which implies (since we assumed that they're coprime) that $E(x,y,z) | P(x,y,z)$ in $\mathbb{F}(y,z)$. Then, by Gauss' Lemma, this implies that $E(x,y,z) | P(x,y,z)$ when considered as polynomials in $\mathbb{F}[y,z]$, the ring of polynomials in $x$.

\end{proof}

Thus, $Q(x,y,z)$ that is degree $d_x$ in $x$ and degree $d_y$ in the skew-$y$ directions exists. The final step in the proof of Theorem~\ref{thm:local-testability} is to analyze the probability of the local views $X$ and $Y$ disagreeing with $Q$.

We have that
\[
    X(x,y,z)E(x,y,z) = Y(x,y,z)E(x,y,z) = Q(x,y,z)E(x,y,z). 
\]
Thus, for any $(b,c)$ for which $E(x, b, c)$ is nonzero, we have that $Q$ agrees with $X$ on the entire row. Since $E$ can be zero on at most $2\delta p^2$ rows (since the combined degree of $y$ and $z$ in $E(x, y, z)$ is at most $2\delta p$ by Lemma~\ref{lem:dx+dy}), this means that \snote{should it be $\delta p^2$?} \rnote{oops, good point. i think you can show $2\delta p^2$: view $E(x,y,z)$ as a $p \times p$ grid in terms of $y \times z$ with entries in $\bbF(x)$. Then the maximal number of zeros is at most max($x$ degree, $y$ degree) $\times p$ by that one theorem which i can't remember the name of ...}
\[
    \Pr_{b,c} [X(x,b,c) \not= Q(x, b, c)] \le 2\delta.
\]
Similarly, we have that $\Pr_{a,c}[Y(a,y,ay+c) \not= Q(a,y,ay+c)] \le 2\delta$. This proves Theorem~\ref{thm:local-testability}. \rnote{rewrote this to be shorter, i think there also might've been an issue with the previous proof?}

\label{sec:locLTC}
\subsection{Local Testability: From Local to Global}

Our main theorem is that if the local codes around an edge have distance $\delta$, and the local codes around a vertex are agreement-testable, then the global code $C$ is agreement-testable, and therefore robustly locally testable.
\begin{theorem}\label{thm:loc2glob}
Let $\epsilon_0,\delta>0$, let $\rho_0(\cdot)$ be a monotone increasing function and assume that $\gamma < \min(\frac \delta 8,\frac {\delta^2}{128}\min(\epsilon_0,\rho_0^{-1}(\frac \delta 4)))$. 
Let $X$ be a two-dimensional bounded-degree $\gamma$-one-sided link expander. 
Suppose we are given codes $C_e\subset \bits^{\T e}$ with relative minimum distance at least $\delta$ for each edge $e\in X(1)$, and let $C_v$ and $C$ be as defined above. 
If $C_v$ is $(\epsilon_0,\rho_0(\cdot))$-agreement-testable for all $v\in X(0)$ then $C$ is $(\epsilon,\rho(\cdot))$-agreement-testable, where 
$\epsilon=\frac {\delta^2}{128}\min(\epsilon_0,\rho_0^{-1}(\frac \delta 4))$ and where $\rho(t) = Dt$ for $D$ the maximal degree of a vertex in $X$.  Namely, given any collection of local views $\sett{z_v\in C_v}{v\in X(0)}$, if 
\[ \alpha(z)=\Pr_{uv\in X(1)}[z_u(\T{uv}) \neq z_v(\T{uv})] < \epsilon\]
then there is some $x\in C$ such that
\[\Pr_v[x|_{\T v} \neq z_v] \leq D\cdot \alpha(z).\]
\end{theorem}
\begin{corollary}
    Under the assumptions above, and assuming that each local code $C_v$ is defined by at most $m_0$ parity checks each looking at most $q_0$ bits, the code $C$ is $(d+2, \epsilon/m_0, \rho'(\cdot))$ locally testable under Definition~\ref{def:LTC-classic} where $\rho'(t) = \rho(m_0 t)$.
\end{corollary}
The corollary follows from a standard conversion from agreement-testability to robust testability, see Section~\ref{sec:agr}, particularly Claims~\ref{claim:agr} and~\ref{claim:loctest-to-loctest}. \inote{which corollary do we want to quote here? regarding LTC with $d+2$ queries?} \rnote{yeah i think so, edited to reflect}

\def\vrej{\rho_0}
\def\erej{\rho_1}
\begin{proof}[Proof of Theorem~\ref{thm:loc2glob}]
Fix $z^0 = \set{z^0_v}$ a collection of local views, such that $\alpha(z^0)<\epsilon$. Run the following local correction algorithm:

\vspace{0.5em}
\noindent
{\bf Local Algorithm:} If there is a vertex $v$ and a choice $z'_v\in C_v$ that reduces $\alpha(z)$ then replace $z_v$ by $z'_v$ and repeat. 
\vspace{0.5em}

Let $z=\set{z_v}$ be the final collection, after the termination of the local algorithm. 
The algorithm must halt in at most $\alpha(z^0)|X(1)|$ steps. At this point, either $\alpha(z)=0$, or $\alpha(z)>0$. In the first case, we will show a nearby codeword. In the second case, we will show that in fact $\alpha(z)>\epsilon$, a contradiction since $\alpha(z) \leq \alpha(z^0) \leq \epsilon$.
\begin{claim} \label{claim:alpha(z)=0}
If $ \alpha(z)=0$ then $z$ corresponds to a codeword $\hat x\in C$ such that $\Pr_{v \in X(0)}[z_v^0 \not= \hat{x}|_{X_{+v}}] \le D \alpha(z^0)$,
where $D$ bounds the maximal degree of a vertex in $X$.
\end{claim}
\begin{proof}
For each triangle $t\in X(2)$ choose arbitrarily a vertex $v\in t$ and set $\hat x(t) = z_{v}(t)$. This choice does not depend on the choice of $v$ because $\alpha(z)=0$ implies that $z_{v}(t) = z_{v'}(t)$ for any $v,v'\in t$.

At every step of the local algorithm, one local view changes. So the number of local views where $\hat x|_{\T{v}}\neq z^0_v$ is at most the number of steps of the algorithm, which is at most $\alpha(z^0)|E|$. So 
\[ \Pr [\hat x|_{\T{v}}\neq z^0_v] \leq \frac{\alpha(z^0)|E|}{|V|} = D\cdot \alpha(z^0) .\] 
\end{proof}
Assume $\alpha(z)>0$, and let
\[R = \sett{uv\in X(1)}{z_u|_{\T{uv}}\neq  z_v|_{\T{uv}}}.\] 
The rest of the proof will show that $R$ has large size.
First, we claim that \begin{equation}\label{eq:updown}
    \Pr_{e \upperrw e' } [ e \in R | e'\in R ] \geq \delta/2.
\end{equation}
where $e \upperrw e'$ is shorthand for the distribution of selecting $e,e'$ from the upper random walk, namely, first choose a random triangle and then choose two distinct edges in it.

Fix an edge $e=uv\in R$. By definition, $(z_u)|_{\T{uv}}\neq (z_v)|_{\T{uv}}$. Since both are codewords in $C_{uv}$, for at least $\delta$ fraction of the triangles $uvw\in \T{uv}$ either  $z_u(uvw)\neq z_w(uvw)$ or $z_v(uvw)\neq z_w(uvw)$ and in particular either $uw\in R$ or $vw\in R$. So the random walk from $uv$ to a triangle $uvw$ and then to $uw$ or $vw$ has probability at least $\delta/2$ of staying in $R$. This establishes \eqref{eq:updown}.
By Lemma \ref{lemma:updown} we further deduce that 
\begin{equation}\label{eq:downup}
    \Pr_{e \lowerrw e' } [ e \in R | e'\in R ] \geq \delta/2-\gamma.
\end{equation}
where $e \lowerrw e'$ is shorthand for the distribution of the lower random walk, namely, selecting a two random edges that intersect on a vertex.

The next step is to focus on the neighborhood of a fixed vertex. Fix $v\in V$. For every neighbor $u$ of $v$ let $y_u = z_u|_{\T{uv}} \in C_{uv}$. This gives us a local view for each neighbor of $v$, which may or may not agree with $(z_v)|_{\T{uv}}$. Let $R(v) = R\cap X_v(1) \supset \{ uw \in X_v(1) ~|~ y_u(uvw) \not= y_w(uvw)\}$. 
We next show that vertices $v$ with small $R(v)$ have few $R$ edges adjacent to them. Here we use the agreement testability of the code $C_v$.
\begin{claim}\label{claim:LTC}
Fix $v\in V$ such that 
$\epsilon_v \eqdef |R(v)|/|X_v(1)| \leq \epsilon_0$. Then $\Pr_{u\in X_v(0)}[uv\in R] \leq \rho_0(\epsilon_v)$.
Contrapositively, if $\Pr_{u\in X_v(0)}[uv\in R]\geq \tau$, then $\epsilon_v \geq \min(\epsilon_0,\rho_0^{-1}(\tau))$.
\end{claim}
\begin{proof}
For every neighbor $u$ of $v$ , $y_u = z_u|_{\T{uv}} \in C_{uv}$ either agrees or disagrees with $z_v$. This is measured by the fraction of $R$ edges touching $v$,
\begin{equation}\label{eq:toR}
     \Pr_{u\in X_v(0)}[uv\in R] = \Pr_{u\in X_v(0)}[z_v|_{\T{uv}} \neq y_u] .
\end{equation}
Note also that for $uw\in X_v(1)$, if $y_u(uvw)\neq y_{w}(uvw)$ then $uw\in R(v)$ and so $uw\in R$. The assumption of the claim implies that $Pr_{uw\in X_v(1)}[y_u(uvw)\neq y_{w}(uvw) ] \leq \epsilon_v\leq \epsilon_0$.
The $(\epsilon_0,\rho_0(\cdot))$ agreement-testability of $C_v$ (see Definition \ref{def:agrtest}) guarantees that in this case there exists a codeword $\hat z_v \in C_v$ such that
\begin{equation}\label{eq:roof1}
\Pr_{u\in X_v(0)}[ \hat z_v|_{\T{uv}}\neq y_u] \leq \rho_0(\epsilon_v)
\end{equation} 
where we have used the monotonicity of $\rho_0(\cdot)$. 

Since the local algorithm halted without changing $z_v$ to $\hat z_v$, we conclude, together with \eqref{eq:toR} and \eqref{eq:roof1}, that
\[
  \Pr_{u\in X_v(0)}[uv\in R] =\Pr_{u\in X_v(0)}[ z_v|_{\T{uv}}\neq y_u] \leq \Pr_{u\in X_v(0)}[ \hat z_v|_{\T{uv}}\neq y_u] \leq \rho_0(\epsilon_v).
\]
To prove the contrapositive notice that if $\epsilon_v \geq \epsilon_0$ it is immediate, and if not, $\tau\leq \rho_0(\epsilon_v)$ which means that $\rho_0^{-1}(\tau)\leq \epsilon_v$ as needed ($\rho_0$ is invertible since it is monotone).
\end{proof}

Let $f=\one_R$. By \eqref{eq:downup},
\begin{equation}
    \label{eq:DD}
\iprod{Df,Df}=\iprod{f,UDf} \ge (\delta/2-\gamma)\norm f^2 .
\end{equation}
Observe also that 
\[ \E_v Df(v) = \E_e f(e) = \E_e f(e)^2 = \snorm{f}.\]
\begin{claim}
    For any $\tau>0$ let $V_\tau = \sett{v\in V}{Df(v)> \tau}$. Then 
    \[\Pr_v(V_\tau) \geq (\delta/2 -\gamma-\tau)\snorm f\]
\end{claim}
\begin{proof} Let $\mu(v)$ denote the probability of a vertex.
\begin{align*}
    (\delta/2-\gamma) \snorm f &\leq \snorm {Df}\\
    &= \sum_{v\in V_\tau} \mu(v) Df(v)^2 + \sum_{v\not\in V_\tau} \mu(v) Df(v)^2 \\
    &\leq  \sum_{v\in V_\tau} \mu(v) \cdot 1 + \sum_{v\not\in V_\tau} \mu(v) Df(v)\cdot\tau\\
    &\leq \sum_{v\in V_\tau} \mu(v) \cdot 1 + \sum_{v\in V} \mu(v) Df(v)\cdot \tau\\
    &=  \Pr_v[V_\tau] + \tau\snorm{f}
\end{align*}
where we have used \eqref{eq:DD} in the first inequality, the definition of $V_\tau$ in the next inequality, and $Df(v)\geq 0$ in the last one.
\end{proof}
Let $\sM = S_{1,0}D$ \inote{changed text description of RW here as well} be the Markov operator corresponding to the random walk that starts at an edge $e$, selects a random vertex such that $e\cup \set v\in X(2)$ (we denote this condition by $e \perp v$), and then chooses a random $e'\ni v$. 
Then, 
\begin{align*}
    \iprod{f,\sM  f}  = \iprod{Df,S_{0,1}f} &= \sum_v \mu(v) Df(v)\cdot S_{0,1} f(v) \\ 
    &= \sum_v \mu(v)( \E_{e\ni v}f(e))( \E_{e'\perp v} f(e'))\\
    &\ge \sum_v \mu(v)( \E_{e\ni v}f(e))\cdot \epsilon_v\\
    &\geq \sum_{v\in V_\tau} \mu(v) \cdot\tau \cdot \min(\epsilon_0,\rho_0^{-1}(\tau))    \\
    &= \tau \min(\epsilon_0,\rho_0^{-1}(\tau)) \Pr[V_\tau]\\
    &\geq \tau \min(\epsilon_0,\rho_0^{-1}(\tau))(\delta/2-\gamma-\tau)\snorm f.
    \end{align*}
where the second inequality is because whenever $v\in V_\tau$, Claim \ref{claim:LTC} implies that $\epsilon_v \geq \min(\epsilon_0,\rho_0^{-1}(\tau))$.

Lemma \ref{lemma:sM} gives a bound of $3\gamma$  on the second largest eigenvalue of $\sM$. We finally apply Lemma \ref{lemma:AC} on the graph corresponding to $\sM$ to deduce that 
\begin{equation}\label{eq:Rbound}
    |R| \ge \big( \tau(\delta/2-\gamma-\tau) \min(\epsilon_0,\rho_0^{-1}(\tau))-3\gamma \big)|X(1)|.
\end{equation}
Choosing for example $\tau=\delta/4$ and as long as $\gamma < \min(\delta/8,\frac {\delta^2}{128}\min(\epsilon_0,\rho_0^{-1}(\frac \delta 4))$ we deduce $|R| > \frac {\delta^2}{128}\min(\epsilon_0,\rho_0^{-1}(\frac \delta 4))\cdot|X(1)| = \epsilon|X(1)|$.
\end{proof}

\subsection{Testability of HDX Codes in Dimensions $k>2$}
In this section we prove a ``trickle-down'' statement for agreement-testability. We show that if the local codes at $i$-links are agreement-testable, then this implies agreement-testability for the codes at $i-1$ links. 
\begin{lemma}\label{lem:testability-trickle-down}
    Let $\delta,\gamma>0$, and let $X$ be a $k$-dimensional $\gamma$-one-sided local expander, and assume that for every $t\in X(k-1)$ we have a code $C_t \subset \set{f:\T t(k)\to\F}$ with minimum relative distance $\delta$. Let $D_i$ denote the maximal number of $i+1$ faces that touch an $i$ face in $X$.
    Define, for every $-1 \le i \le k-2$ face $s\in X(i)$, the code \[C_s = \sett{f:\T s(k)\to\F}{f|_{\T t}\in C_t\;\forall  t\in \T s(k-1)},\]
    and assume that for $s \in X(i)$ the code $C_s$ has minimum relative distance $\delta_i$.
    If there is $\epsilon_{k-2}>0$ and a monotone increasing function $\rho_{k-2}(\cdot)$ such that for every $t\in X(k-2)$ the code $C_t$ is $(\epsilon_{k-2},\rho_{k-2}(\cdot))$-agreement testable, then for every $-1\leq i<k-2$ and every $s\in X(i)$ the code $C_s$ is $(\epsilon_{i},\rho_{i}(\cdot))$-agreement testable with 
    \[
        \epsilon_i = \frac{\delta_{i+1}^2}{128} \min \left( \epsilon_{i+1}, \rho_{i+1}^{-1} \left( \frac{\delta_{i+1}}{4} \right) \right), \qquad \rho_i(t) = D_{i+1} t
    \]
    as long as $\gamma < \min\left( \frac{\delta_{i+1}}{8}, \frac{\delta_{i+1}^2}{128} \min \left( \epsilon_{i+1}, \rho_{i+1}^{-1}(\delta_{i+1}/4) \right) \right)$ for all $i$.
    In particular, the code $C_{\emptyset} \subset \set{f:X(k)\to\F}$ is $(\epsilon_{-1},\rho_{-1}(\cdot))$-agreement testable. 
\end{lemma}
\begin{proof}
    The proof is very similar to the proof of \pref{thm:loc2glob}. Fix $s\in X(i)$. Our goal is to prove that $C_s$ is testable assuming that for all $v\in X_s(0)$, $C_{s\cup\set v}$ is testable. By bijecting $\T s$ to $X_s$ (mapping each $t\in \T s$ to $t\setminus s$), we can move to the case where $s=\emptyset$, and our codes sit on the faces of dimension $k-i-1$, namely $C_v\subseteq \bits^{\T v(k-i-1)}$ and $C_\phi \subseteq \bits^{X(k-i-1)}$. The complex $X_s$ is a $\gamma$ high dimensional expander, and each vertex $v$ touches at most $D_{i+1}$ edges \snote{Since $s\in X(i)$, each $v\in X_s(0)$ touches $D_{i+1}$ edges}. We also let $\delta_{i+1}$ denote the distance of the code $C_{s \cup \{v\}}$ (see Lemma~\ref{lemma:hd-distance} for a calculation). 

    When $i=k-3$, this was done in the previous section. So the only difference is that now $C_v$ itself is a $(k-i-1)$-dimensional HDX code for possibly $k-i-1>1$. In fact we will see that this makes almost no difference.

    We start with $z^0 = \{ z^0_v \in C_v \}_{v \in X(0)}$ a collection of local views, and define $\alpha(z^0) = \Pr_{uv \in X(1)} [z_u(X_{+uv}) \not= z_v(X_{+uv})]$. Same as in the proof of Theorem~\ref{thm:loc2glob}, we run the local algorithm, replacing $z_v$ with $z'_v \in C_v$ whenever it reduces $\alpha(z)$, until no more changes can be made. Let $z = \{ z_v \}_{v \in X(0)}$ be the final collection.

    Then, by following the same steps as in the proof of Theorem~\ref{thm:loc2glob}, we have the following:
    \begin{itemize}
    \item 
        If $\alpha(z) = 0$, then just as in Claim~\ref{claim:alpha(z)=0} we have that $z$ corresponds to a codeword $\hat{x} \in C$ satisfying $\Pr_{v \in X(0)} [z^0_v \not= \hat{x}|_{X_{+v}}] \le D_{i+1} \alpha(z^0)$. \snote{Should be $D_{i+1}$}
    \item 
        If $\alpha(z) \not= 0$, then the goal is to show that there must have been many edge disagreements to begin with. Define $R = \{ uv \in X(1) ~|~ z_u|_{X_{+uv}} \not= z_v|_{X_{+uv}} \}$. Then as long as $\gamma < \min\left( \frac{\delta_{i+1}}{8}, \frac{\delta_{i+1}^2}{128} \min \left( \epsilon_{i+1}, \rho_{i+1}^{-1}(\delta_{i+1}/4) \right) \right)$, it holds that 
        \[
            |R| \ge \frac{\delta_{i+1}^2}{128} \min \left( \epsilon_{i+1}, \rho_{i+1}^{-1} \left( \frac{\delta_{i+1}}{4} \right)\right) \cdot |X(1)| = \epsilon_{i} |X(1)|.
        \]
        Therefore, the number of edge disagreements to begin with in the original ensemble $z^0$ must have been at least $\epsilon_i |X(1)|$ also.
    \end{itemize}

\end{proof}\label{sec:globLTC}

\section{Codes in Higher Dimensions}\label{sec:higher}
The coset complexes considered here have a higher dimensional version as follows. Fix $k>2$ and define $H_i = \sett{h_i(\alpha)}{\alpha \in \bbF}$ where $h_i(\alpha) = e_{i,i+1}(\alpha t) + I_{k+1}$ 
and let $K_i = span(H_j : j\neq i)$ and $G = span (H_1,\ldots,H_{k+1})$. 
Let $X$ be the $k$-dimensional complex $X[G;K_1,\ldots,K_{k+1}]$.

We have the following properties of $X$:
\begin{itemize}
\item 
    $X$ is a $\gamma$-expander, where $\gamma =  \frac{1}{\sqrt{\card{\F}} - (k-1)}$. We use the trickle down theorem together with the fact that the link of any $t\in X(d-2)$ is either a $\frac{1}{\sqrt{\card{\F}}}$-expander (\cref{claim:eigenvalue}) or a complete bipartite graph (justification in proof of \cref{lem:k-2-face}).
\item 
    For any $t \in X(i)$, the number of $(i+1)$-faces that touch $t$ is at most $D_i = (k-i) \cdot |\bbF|^{k-i-1}$. The reason is that $t \in X(i)$ corresponds to a coset of the group generated by $k-i$ subgroups $H_{j_1}, \dots, H_{j_{k-i}}$. Each $(i+1)$ face that touches $t$ is a coset of the group generated by $k-i-1$ of those subgroups. For each such collection of $k-i-1$ subgroups, there are at most $|\bbF|^{k-i-1}$ cosets of the resulting group contained within $t$.
\end{itemize}

Define like before an embedding of the group elements of $X$ into a vector space
\[ G\to R^{(k+1)^2} \cong \F^{n(k+1)^2}\]
so that the elements in each $gH_i$ embed to an entire affine line. Here we will be working with fields $\bbF = \bbF_p$ that have prime order. Fix a degree parameter $k\gamma |\bbF| < d< \card \bbF/4$ \rnote{changed conditions on $d$} and define for every $t\in X(k-1)$ the code $C_t$ to be the Reed-Solomon code $RS(\card\F,d)$. Further define, for every $i < k-1$ and every face $w\in X(i)$, 
\[C_w = \sett{f:\T w(k)\to\F~}{f|_{\T t}\in C_t\;\forall  t\in \T w(k-1)}.\]
It is immediate that the code $C=\calC^k[X,\set{C_v}_{v\in X(0)}]\subseteq \F^{X(k)}$ is an HDX code as defined in Section \ref{sec:agr}. Furthermore, by Lemma~\ref{lemma:hd-distance}, for any $w \in X(i)$, the code $C_w$ has distance $\ge \delta_i = \prod_{j = 0}^{k-1-i} (\delta - j\gamma)$.

\begin{theorem} \label{thm:loc2glob-highdim}
    Let $\bbF$ be a field of prime order, and let $X = X[G; K_1, \dots, K_{k+1}]$ be a $k$-dimensional complex. If $k\gamma |\bbF| < d < \left( \frac14 - \frac{5}{64} \delta_{k-2} \right) |\bbF|$ and \inote{there was some condition below on $\delta_{k-2}$ and it seems that's it} \rnote{don't know what other conditions we need?}, then for every $i< k-1$ and every $w\in X(i)$, the code $C_w$ is $(\epsilon_i, \rho_i)$-agreement testable, where
    \[
        \epsilon_{k-2} = \left( \frac{p - 4d}{5p} \right)^3 ~~~\text{and}~~~ \rho_{k-2}(\cdot) = 4(\cdot)^{1/3},
    \]
    and for $-1 \le i < k-2$,
    \[
        \epsilon_i = \frac{\delta_{k-2}^3 \cdot \prod_{j = i+1}^{k-2} \delta_j^2}{2^{5} \cdot 2^{7(k-i-1)} \cdot (k+1) \cdot |\bbF|^k} ~~~\text{and}~~~ \rho_i = D_{i+1} \cdot (\cdot),
    \]
    where $\delta_i = \prod_{j = 0}^{k-1-i} (\delta - j\gamma)$ and $D_i = (k-i) \cdot |\bbF|^{k-i-1}$.

    In particular, the code $C_\phi \subset \set{f:X(k)\to\F}$ is $(\epsilon_{-1},\rho_{-1}(\cdot))$-agreement testable.
\end{theorem}

The rest of this section is dedicated to proving Theorem~\ref{thm:loc2glob-highdim}. Suppose first that $i=k-2$ and $s\in X(k-2)$. Recall that $X$ is $(k+1)$-partite, and denote $color(s)\subset[k+1]$ the set of colors of $s$. For $k > 2$ we observe two kinds of links $X_s$, and therefore two kinds of codes $C_s$, depending on the color of $s$.
\begin{lemma} \label{lem:k-2-face}
    Let $s \in X(k-2)$. Let $color(s) = [k+1]\setminus \set{i,j}$. If $|i-j|\equiv1 \mod (k+1)$ then $C_s$ is isomorphic to $C_{d,d}$ as defined in \eqref{eq:def:vertexcode}.
    Otherwise, $C_s \cong RS(\card\F,r)^{\otimes 2}$.

    In both cases $C_s$ is $\left( \left( \frac{p - 4d}{5p} \right)^3, 4 (\cdot)^{1/3} \right)$-agreement testable.
\end{lemma}
\begin{proof}
    When $\card{i-j}>1\mod (k+1)$ the subgroups $H_i$ and $H_j$ commute since $[h_i(\alpha), h_j(\beta)] = 0$. Then $span(H_i,H_j) \cong \F^2$, and $C_s \cong RS(\card\F,d)^{\otimes 2}$. Also, $\left(\left( \frac{p-2d}{2p} \right)^2, 2(\cdot) \right)$-agreement testability was proven in \cite{PolishchukS94} as long as $d< \card\F/2$.

    When $\card {i-j}=1\mod (k+1)$, we can ignore a large part of the matrices (which is identity). For example, if $i=1$ and $j=2$ then we can restrict attention to the first $3$ rows and columns of the matrices. Thus, we are back in the $k=2$ case, with $span(H_i,H_j)$ isomorphic to the group $K_{6-i-j}$ defined in Section \ref{subsec:X}, and the code $C_s$ is isomorphic to $C_{d,d}$ as defined in \eqref{eq:def:vertexcode}. In this case, by Corollary~\ref{cor:agr-test-local} $C_s$ is $\left( \left( \frac{p - 4d}{5p} \right)^3, 4 (\cdot)^{1/3} \right)$-agreement testable when $d < |\bbF|/4$.
\end{proof}

Moving to $i<k-2$, our proof relies on reverse induction, deducing agreement testability of the level $i$ codes from agreement testability of the level $i+1$ codes.

\begin{lemma}\label{lem:ind}
%
    If $d < \left( \frac14 - \frac{5}{64} \delta_{k-2} \right) |\bbF|$, then for every $-1\leq i<k-2$ and every $s\in X(i)$ the code $C_s$ is $(\epsilon_{i},\rho_{i}(\cdot))$-agreement testable with 
    \[
        \epsilon_i = \frac{\delta_{k-2}^3 \cdot \prod_{j = i+1}^{k-2} \delta_j^2}{2^{5} \cdot 2^{7(k-i-1)} \cdot (k+1) \cdot |\bbF|^k} ~~\text{and}~~ \rho_i(x) = D_{i+1} \cdot x,
    \]
    where $\delta_i = \prod_{j = 0}^{k-1-i} (\delta - j\gamma)$ and $D_i = (k-i) \cdot |\bbF|^{k-i-1}$.
\end{lemma}

\begin{proof}
    We have from Lemma~\ref{lem:k-2-face} that for any $s \in X(k-2)$ $C_s$ is $(\epsilon_{k-2}, \rho_{k-2}(\cdot))$-agreement testable, where $\epsilon_{k-2} = \left( \frac{p - 4d}{5p} \right)^3$ and $\rho_{k-2}(\cdot) = 4 (\cdot)^{1/3}$. Then for $s \in X(k-3)$, we have by Lemma~\ref{lem:testability-trickle-down} that $C_s$ is $(\epsilon_{k-3}, \rho_{k-3}(\cdot))$-agreement testable where $\rho_{k-3}(x) = D_{k-2} \cdot x$ and
    \begin{align*}
        \epsilon_{k-3} 
        &= \frac{\delta_{k-2}^2}{128} \min \left( \epsilon_{k-2}, \rho_{k-2}^{-1} \left( \delta_{k-2}/4 \right) \right) \\
        &= \frac{\delta_{k-2}^2}{128} \min \left( \left( \frac{p - 4d}{5p} \right)^3, \left( \delta_{k-2}/16 \right)^3 \right) \\
        &= \frac{\delta^5_{k-2}}{2^{19}} \\
        &\ge \frac{\delta_{k-2}^5}{2^{19} D_{-1}},
    \end{align*}
    where the third equality holds whenever $\frac{p-4d}{5p} > \frac{\delta_{k-2}}{16}$, which happens whenever $d < \left( \frac14 - \frac{5}{64} \delta_{k-2} \right) |\bbF|$. \inote{strange that we want $\delta_{k-2}$ to be small here; anyhow, need to add this condition upstairs}

    In general, for $s \in X(i)$, the code $C_s$ is $(\epsilon_i, \rho_i(\cdot))$-agreement testable where $\rho_i(x) = D_{i+1} \cdot x$ and
    \begin{align*}
        \epsilon_i &= \frac{\delta_{i+1}^2}{128} \min \left( \epsilon_{i+1}, \frac{\delta_{i+1}}{4D_{i+2}} \right) 
        \ge \frac{\delta_{k-2}^3 \cdot \prod_{j = i+1}^{k-2} \delta_j^2}{2^{5} \cdot 2^{7(k-i-1)} \cdot D_{-1}} 
    \end{align*}
    since 
    \[
        \epsilon_{i+1} 
        = \frac{\delta_{k-2}^3 \cdot \prod_{j = i+2}^{k-2} \delta_j^2}{2^5 \cdot 2^{k-i-2} \cdot D_{-1}} 
        < \frac{\delta_{i+1}}{4D_{i+2}}. 
    \]

    Plugging in $D_{-1} = (k+1)|\bbF|^k$ finishes the calculation.

    \inote{yay}
\end{proof}

\bibliographystyle{alpha}
\bibliography{refs}

\newcommand{\etalchar}[1]{$^{#1}$}
\begin{thebibliography}{DDHR20b}

\bibitem[ALM{\etalchar{+}}98]{ALMSS}
S.~Arora, C.~Lund, R.~Motwani, M.~Sudan, and M.~Szegedy.
\newblock Proof verification and intractability of approximation problems.
\newblock {\em Journal of the ACM}, 45(3):501--555, 1998.

\bibitem[AS98]{AS}
S.~Arora and S.~Safra.
\newblock Probabilistic checking of proofs: A new characterization of {NP}.
\newblock {\em Journal of the ACM}, 45(1):70--122, 1998.

\bibitem[DDFH18]{DiksteinDFH2018}
Yotam Dikstein, Irit Dinur, Yuval Filmus, and Prahladh Harsha.
\newblock Boolean function analysis on high-dimensional expanders.
\newblock In {\em Proc.\ $20$th International Workshop on Randomization and
  Computation (RANDOM)}, volume 116, pages 37:1--37:21, Princeton, NJ, 2018.
  RANDOM/APPROX.

\bibitem[DDHR20a]{DiksteinDHR20}
Yotam Dikstein, Irit Dinur, Prahladh Harsha, and Noga Ron{-}Zewi.
\newblock Locally testable codes via high-dimensional expanders.
\newblock {\em Electron. Colloquium Comput. Complex.}, {TR20-072}, 2020.

\bibitem[DDHR20b]{DDHR}
Yotam Dikstein, Irit Dinur, Prahladh Harsha, and Noga {Ron-Zewi}.
\newblock Locally testable codes via high-dimensional expanders.
\newblock {\em arXiv preprint arXiv:2005.01045}, 2020.

\bibitem[DEL{\etalchar{+}}22a]{DinurELLM2022}
Irit Dinur, Shai Evra, Ron Livne, Alexander Lubotzky, and Shahar Mozes.
\newblock Locally testable codes with constant rate, distance, and locality.
\newblock In Stefano Leonardi and Anupam Gupta, editors, {\em {STOC} '22: 54th
  Annual {ACM} {SIGACT} Symposium on Theory of Computing, Rome, Italy, June 20
  - 24, 2022}, pages 357--374. {ACM}, 2022.

\bibitem[DEL{\etalchar{+}}22b]{DELLM}
Irit Dinur, Shai Evra, Ron Livne, Alexander Lubotzky, and Shahar Mozes.
\newblock Locally testable codes with constant rate, distance, and locality.
\newblock In {\em Proceedings of the 54th Annual ACM SIGACT Symposium on Theory
  of Computing}, pages 357--374, 2022.

\bibitem[DHKRz19]{DinurHKR2018}
Irit Dinur, Prahladh Harha, Tali Kaufman, and Noga Ron-zewi.
\newblock From local to robust testing via agreement testing.
\newblock In {\em Proceedings of the 11th Innovations in Theoretical Computer
  Science (ITCS) Conference}, 2019.

\bibitem[DK17]{DinurK2017}
Irit Dinur and Tali Kaufman.
\newblock High dimensional expanders imply agreement expanders.
\newblock In {\em Proc.\ $58$th IEEE Symp.\ on Foundations of Comp.\ Science
  (FOCS)}, pages 974--985, 2017.

\bibitem[EK16]{EvraK2016}
Shai Evra and Tali Kaufman.
\newblock Bounded degree cosystolic expanders of every dimension.
\newblock In {\em Proc.\ $48$th ACM Symp.\ on Theory of Computing (STOC)},
  pages 36--48, 2016.

\bibitem[FGW17]{FGW2017}
S.~Luna Frank{-}Fischer, Venkatesan Guruswami, and Mary Wootters.
\newblock Locality via partially lifted codes.
\newblock {\em CoRR}, abs/1704.08627, 2017.

\bibitem[FK22]{firstK2022good}
Uriya~A First and Tali Kaufman.
\newblock On good $2 $-query locally testable codes from sheaves on high
  dimensional expanders.
\newblock {\em arXiv preprint arXiv:2208.01778}, 2022.

\bibitem[GKS13]{GuoKS13}
Alan Guo, Swastik Kopparty, and Madhu Sudan.
\newblock New affine-invariant codes from lifting.
\newblock In Robert~D. Kleinberg, editor, {\em Innovations in Theoretical
  Computer Science, {ITCS} '13, Berkeley, CA, USA, January 9-12, 2013}, pages
  529--540. {ACM}, 2013.

\bibitem[Gol23]{Gol23}
Louis Golowich.
\newblock From grassmannian to simplicial high-dimensional expanders.
\newblock {\em CoRR}, abs/2305.02512, 2023.

\bibitem[GV85]{GesselV85}
Ira Gessel and Gérard Viennot.
\newblock Binomial determinants, paths, and hook length formulae.
\newblock {\em Advances in Mathematics}, 58(3):300--321, 1985.

\bibitem[HS19]{Harsha_Saptharishi_2019}
Prahladh Harsha and Ramprasad Saptharishi.
\newblock A note on the elementary {HDX} construction of kaufman-oppenheim.
\newblock {\em arXiv: Discrete Mathematics}, Dec 2019.

\bibitem[KKL14]{KaufmanKL2014}
Tali Kaufman, David Kazhdan, and Alexander Lubotzky.
\newblock Ramanujan complexes and bounded degree topological expanders.
\newblock In {\em Proc.\ $55$th IEEE Symp.\ on Foundations of Comp.\ Science
  (FOCS)}, pages 484--493, 2014.

\bibitem[KL12]{KaufmanL12}
Tali Kaufman and Alexander Lubotzky.
\newblock Edge transitive ramanujan graphs and symmetric {LDPC} good codes.
\newblock In Howard~J. Karloff and Toniann Pitassi, editors, {\em Proceedings
  of the 44th Symposium on Theory of Computing Conference, {STOC} 2012, New
  York, NY, USA, May 19 - 22, 2012}, pages 359--366. {ACM}, 2012.

\bibitem[KMS18]{KMS}
Subhash Khot, Dor Minzer, and Muli Safra.
\newblock Pseudorandom sets in grassmann graph have near-perfect expansion.
\newblock In Mikkel Thorup, editor, {\em 59th {IEEE} Annual Symposium on
  Foundations of Computer Science, {FOCS} 2018, Paris, France, October 7-9,
  2018}, pages 592--601. {IEEE} Computer Society, 2018.

\bibitem[KO18]{KaufmanO181}
Tali Kaufman and Izhar Oppenheim.
\newblock Construction of new local spectral high dimensional expanders.
\newblock In {\em Proceedings of the 50th Annual {ACM} {SIGACT} Symposium on
  Theory of Computing, {STOC} 2018, Los Angeles, CA, USA, June 25-29, 2018},
  pages 773--786, 2018.

\bibitem[KO20]{KaufmanO-RW20}
Tali Kaufman and Izhar Oppenheim.
\newblock High order random walks: Beyond spectral gap.
\newblock {\em Comb.}, 40(2):245--281, 2020.

\bibitem[KT23]{KaufmanT23}
Tali Kaufman and Ran~J. Tessler.
\newblock Garland's technique for posets and high dimensional grassmannian
  expanders.
\newblock In Yael~Tauman Kalai, editor, {\em 14th Innovations in Theoretical
  Computer Science Conference, {ITCS} 2023, January 10-13, 2023, MIT,
  Cambridge, Massachusetts, {USA}}, volume 251 of {\em LIPIcs}, pages
  78:1--78:22. Schloss Dagstuhl - Leibniz-Zentrum f{\"{u}}r Informatik, 2023.

\bibitem[KW16]{KaufmanW16}
Tali Kaufman and Avi Wigderson.
\newblock Symmetric {LDPC} codes and local testing.
\newblock {\em Comb.}, 36(1):91--120, 2016.

\bibitem[LFKN92]{LFKN}
C.~Lund, L.~Fortnow, H.~Karloff, and N.~Nisan.
\newblock Algebraic methods for interactive proof systems.
\newblock {\em Journal of the ACM}, 39(4):859--868, October 1992.

\bibitem[LSV05a]{LubotzkySV2005b}
Alexander Lubotzky, Beth Samuels, and Uzi Vishne.
\newblock Explicit constructions of {R}amanujan complexes of type
  $\tilde{A_d}$.
\newblock {\em European J. Combin.}, 26(6):965--993, 2005.

\bibitem[LSV05b]{LubotzkySV2005a}
Alexander Lubotzky, Beth Samuels, and Uzi Vishne.
\newblock Ramanujan complexes of type $\tilde{A_d}$.
\newblock {\em Israel J.\ Math.}, 149(1):267--299, 2005.

\bibitem[Mei13]{Meir13}
Or~Meir.
\newblock {IP} = {PSPACE} using error-correcting codes.
\newblock {\em {SIAM} J. Comput.}, 42(1):380--403, 2013.

\bibitem[Mes18]{meshulam2018graph}
Roy Meshulam.
\newblock Graph codes and local systems, 2018.

\bibitem[MR08]{MR-LDT}
Dana Moshkovitz and Ran Raz.
\newblock Sub-constant error low degree test of almost-linear size.
\newblock {\em SIAM J. Computing}, 38(1):140--180, 2008.

\bibitem[OP22]{ODonnellP2022}
Ryan O'Donnell and Kevin Pratt.
\newblock High-dimensional expanders from chevalley groups.
\newblock In Shachar Lovett, editor, {\em 37th Computational Complexity
  Conference, {CCC} 2022, July 20-23, 2022, Philadelphia, PA, {USA}}, volume
  234 of {\em LIPIcs}, pages 18:1--18:26. Schloss Dagstuhl - Leibniz-Zentrum
  f{\"{u}}r Informatik, 2022.

\bibitem[PK22]{PanteleevK22}
Pavel Panteleev and Gleb Kalachev.
\newblock Asymptotically good quantum and locally testable classical {LDPC}
  codes.
\newblock In Stefano Leonardi and Anupam Gupta, editors, {\em {STOC} '22: 54th
  Annual {ACM} {SIGACT} Symposium on Theory of Computing, Rome, Italy, June 20
  - 24, 2022}, pages 375--388. {ACM}, 2022.

\bibitem[PS94]{PolishchukS94}
Alexander Polishchuk and Daniel~A. Spielman.
\newblock Nearly-linear size holographic proofs.
\newblock In Frank~Thomson Leighton and Michael~T. Goodrich, editors, {\em
  Proceedings of the Twenty-Sixth Annual {ACM} Symposium on Theory of
  Computing, 23-25 May 1994, Montr{\'{e}}al, Qu{\'{e}}bec, Canada}, pages
  194--203. {ACM}, 1994.

\bibitem[RS97]{RaSa}
R.~Raz and S.~Safra.
\newblock A sub-constant error-probability low-degree test, and a sub-constant
  error-probability {PCP} characterization of {NP}.
\newblock pages 475--484, 1997.

\bibitem[Sha92]{Sha:IP:PSPACE}
A.~Shamir.
\newblock {IP = PSPACE}.
\newblock {\em Journal of the ACM}, 39(4):869--877, October 1992.
\newblock Prelim. version in 1990 FOCS, pages 11--15.

\bibitem[SS96]{SipserS96}
Michael Sipser and Daniel~A. Spielman.
\newblock Expander codes.
\newblock {\em {IEEE} Trans. Inf. Theory}, 42(6):1710--1722, 1996.

\end{thebibliography}


\newcommand{\etalchar}[1]{$^{#1}$}
\begin{thebibliography}{DEL{\etalchar{+}}22}

\bibitem[DEL{\etalchar{+}}22]{DinurELLM2022}
Irit Dinur, Shai Evra, Ron Livne, Alexander Lubotzky, and Shahar Mozes.
\newblock Locally testable codes with constant rate, distance, and locality.
\newblock In Stefano Leonardi and Anupam Gupta, editors, {\em {STOC} '22: 54th Annual {ACM} {SIGACT} Symposium on Theory of Computing, Rome, Italy, June 20 - 24, 2022}, pages 357--374. {ACM}, 2022.

\bibitem[EK16]{EvraK2016}
Shai Evra and Tali Kaufman.
\newblock Bounded degree cosystolic expanders of every dimension.
\newblock In {\em Proc.\ $48$th ACM Symp.\ on Theory of Computing (STOC)}, pages 36--48, 2016.

\bibitem[GV85]{GesselV85}
Ira Gessel and Gérard Viennot.
\newblock Binomial determinants, paths, and hook length formulae.
\newblock {\em Advances in Mathematics}, 58(3):300--321, 1985.

\bibitem[HS19]{Harsha_Saptharishi_2019}
Prahladh Harsha and Ramprasad Saptharishi.
\newblock A note on the elementary {HDX} construction of kaufman-oppenheim.
\newblock {\em arXiv: Discrete Mathematics}, Dec 2019.

\bibitem[KKL14]{KaufmanKL2014}
Tali Kaufman, David Kazhdan, and Alexander Lubotzky.
\newblock Ramanujan complexes and bounded degree topological expanders.
\newblock In {\em Proc.\ $55$th IEEE Symp.\ on Foundations of Comp.\ Science (FOCS)}, pages 484--493, 2014.

\bibitem[KO18]{KaufmanO181}
Tali Kaufman and Izhar Oppenheim.
\newblock Construction of new local spectral high dimensional expanders.
\newblock In {\em Proceedings of the 50th Annual {ACM} {SIGACT} Symposium on Theory of Computing, {STOC} 2018, Los Angeles, CA, USA, June 25-29, 2018}, pages 773--786, 2018.

\bibitem[OP22]{ODonnellP2022}
Ryan O'Donnell and Kevin Pratt.
\newblock High-dimensional expanders from chevalley groups.
\newblock In Shachar Lovett, editor, {\em 37th Computational Complexity Conference, {CCC} 2022, July 20-23, 2022, Philadelphia, PA, {USA}}, volume 234 of {\em LIPIcs}, pages 18:1--18:26. Schloss Dagstuhl - Leibniz-Zentrum f{\"{u}}r Informatik, 2022.

\bibitem[PK22]{PanteleevK22}
Pavel Panteleev and Gleb Kalachev.
\newblock Asymptotically good quantum and locally testable classical {LDPC} codes.
\newblock In Stefano Leonardi and Anupam Gupta, editors, {\em {STOC} '22: 54th Annual {ACM} {SIGACT} Symposium on Theory of Computing, Rome, Italy, June 20 - 24, 2022}, pages 375--388. {ACM}, 2022.

\end{thebibliography}

\end{document}